\documentclass{article}
\usepackage{subfigure}
\usepackage{afterpage}
\usepackage{url}
\usepackage{graphicx,algorithmic,algorithm,bbm}
\usepackage{amsmath,amssymb,latexsym,float,epsfig,subfigure}
\usepackage{hyperref}

\usepackage{amsthm}
\usepackage{fullpage}
\newenvironment{citemize}{
\begin{list}{\labelitemi}{\leftmargin=1.5em}
  \setlength{\itemsep}{1.5pt}
  \setlength{\parskip}{0pt}
  \setlength{\parsep}{0pt}}{\end{list}
}

\DeclareMathAlphabet{\mathpzc}{OT1}{pzc}{m}{it}
\newtheorem{theorem}{Theorem}[section]

\newcommand{\Real}{\mathbb{R}}

\newcommand{\dims}{\mathpzc{D}}

\newcommand{\dist}{\mathbf{d}}

\newcommand{\xvec}{\mathbf{x}}

\newcommand{\ip}[2]{\ensuremath{\langle #1, #2 \rangle}}
\newcommand{\norm}[1]{\ensuremath{\left\| #1 \right\|}}
\newcommand{\normsq}[1]{\ensuremath{\left\| #1 \right\|^2}}
\newcommand{\twonorm}[1]{\ensuremath{\left\| #1 \right\|_2}}
\newcommand{\twonormsq}[1]{\ensuremath{\left\| #1 \right\|_2^2}}
\newcommand{\ball}[2]{\ensuremath{\mathcal{B}_{#1}^{#2}}}
\newcommand{\cone}[2]{\ensuremath{\mathcal{C}_{#1}^{#2}}}

\begin{document}

\title{Maximum Inner-Product Search using Tree Data-structures}

\author{Parikshit Ram \and Alexander G. Gray}
\maketitle

\begin{abstract}
The problem of {\em efficiently} finding the best match for a query in a given set with respect to the Euclidean distance or the cosine similarity has been extensively studied in literature. However, a closely related problem of efficiently finding the best match with respect to the inner product has never been explored in the general setting to the best of our knowledge. In this paper we consider this general problem and contrast it with the existing best-match algorithms. First, we propose a general branch-and-bound algorithm using a tree data structure. Subsequently, we present a dual-tree algorithm for the case where there are multiple queries. Finally we present a new data structure for increasing the efficiency of the dual-tree algorithm. These branch-and-bound algorithms involve novel bounds suited for the purpose of best-matching with inner products. We evaluate our proposed algorithms on a variety of data sets from various applications, and exhibit up to five orders of magnitude improvement in query time over the naive search technique.
\end{abstract}

\section{Introduction}
\label{sec:intro}
%\begin{citemize}
%\item[**] Introduce the problem
%\item[*] Contrast to existing work 
%\item[**] Why is it harder than existing work?
%\item[**] Applications (why do we care?)
%\item[**] Why trees? (Since trees are known to be bad indexing schemes in higher dimensions)
%\end{citemize}
%
In this paper, we consider the problem of {\em efficiently} finding the best-match for a query from a given set of points with respect to the inner-product similarity. Formally, we consider the following problem:

\noindent
\textbf{Problem.} For a given set of $N$ points $S \subset \Real^\dims$ and a query $q\in \Real^\dims$, efficiently find a point $p\in S$ such that:
\begin{equation} \label{eq:maxip}
\ip{q}{p} = \max_{r \in S} \ip{q}{r}.
\end{equation}
We call this the problem of {\em maximum inner-product search}. The focus of this paper is to improve the efficiency of this search. An alternate formulation of the above problem in terms of a vector and matrix multiplication is as following: 

\noindent
\textbf{Problem.} For a given vector $w \in \Real^\dims$ and a matrix $M \in \Real^{\dims \times N}$, efficiently compute the following: 
\begin{equation} \label{eq:maxip_alt}
\norm{w^T M}_\infty = \max (w^T M).
\end{equation}

This problem appears to be very similar much existing work in literature. Efficiently finding the best match with respect to the Euclidean (or more generally $L_p$) distance is the widely studied problem of fast nearest-neighbor search in metric spaces \cite{clarkson2006nearest}. Efficient retrieval of the best match with respect to the cosine similarity is the extensively researched in the field of  text mining and information retrieval \cite{bayardo2007scaling}. But as we will explain in the next section, the maximum inner-product search is not only different from these aforementioned tasks, but also arguably harder.
\subsection{Applications}
%\begin{citemize}
%\item Retrieval of recommendations in the Matrix-Factorization framework.
%\item Document retrieval
%\item Any place you need to compute the $L_\infty$ norm repeatedly (hence need to efficiency)
%\item[*] Any optimization of this form
%\item Solving the problem: Given a set of points $S$, a Mercer-kernel function $K(\cdot,\cdot)$, for a query $q$, find the point $p\in S$ such that
%\begin{equation}
%\label{eq:max-kernel-operation}
%p = \arg \max_{r \in S} K(q,r)
%\end{equation}
%Known as the max-kernel operation, it is widely used in belief propagation (find citations), images matching (computer vision \cite{kulis2009kernelized}, (cite others)). Also talk about Rahimi's thing. 
%\end{citemize}
%
%An important question to ask is whether this problem is worth addressing. In this subsection, we present some potential applications from many different fields of machine-learning and data mining.
%
An obvious application of maximum inner-product search stems out of the widely successful matrix-factorization framework in recommender system challenges like the ``Netflix prize'' \cite{koren2009matrix, koren2009bellkor, bell2007lessons}. The matrix-factorization task obtains accurate representation of the available data in terms of user vectors and items vectors (examples for items would be movies or music). In this setting, the preference of a user for an item is the inner-product between the corresponding user's vector and the item's vector. The efficient retrieval of recommendations for a user is equivalent to the problem in equation \ref{eq:maxip} with the user as the query and the items as the reference set. For the challenges, linear scan of the items are usually employed to find the best recommendations. 
%But for real world systems with large number of users and items, linear scan is not feasible, hence limiting the scope of matrix-factorization techniques to smaller systems. 
An efficient search algorithm would make the retrieval of recommendations in the matrix-factorization framework scalable to real world systems.

The usual document retrieval tasks use the cosine-similarity to match documents. However, in certain setting \cite{deerwester1990indexing}, the documents are represented as (not necessarily normalized) vectors and the inner-product between these vectors represent the similarity between the documents. In this case, unless the vectors are normalized to have the same length, document matching using the cosine-similarity \cite{bayardo2007scaling} might make the algorithm scalable at the cost of returning inaccurate solutions since the inner-product is not the same as the cosine-similarity (we will explain this more elaborately in the section \ref{sec:related}). 

There is a similar problem known as the the max-kernel operation: for a given set of points $S$ and a query $q$ and a kernel-function $\mathcal{K}(\cdot,\cdot)$, the task is to find the point $p\in S$ with the maximum value of $\mathcal{K}(q,p)$ over the set $S$. This problem is widely used in maximum-a-posteriori inference \cite{klaas2005fast} in machine learning, and for the task of image matching \cite{kulis2009kernelized} in computer vision. If the kernel function can be explicitly represented in the form a function $\varphi(\cdot)$ such that $\mathcal{K}(q,p) = \ip{\varphi(q)}{\varphi(r)}$, then this problem reduces to the problem in equation \ref{eq:maxip} after all the points in the set $S$ and the query $q$ is transformed into the $\varphi$-space.
%\footnote{The problems are still equivalent if the $\phi$-space is infinite-dimensional. However, in this paper, we address the setting where there are finite number of dimensions.}. 
%
%Finally, the alternate formulation of the problem (equation \ref{eq:maxip_alt}) is part of any linear algebra library implementing the $L_\infty$ norm. If this quantity is evaluated multiple times for different vectors $w$ and a fixed matrix $M$, then an efficient solution to this task would boost the efficiency of the library. (Talk to Nishant and Krishna about possible sparse coding and multi-label learning problems.)
%\vspace{-0.1in}
\subsection{This Paper}
%\vspace{-0.05in}
%\begin{citemize}
%\item Contrast this problem to existing best-match problems and existing techniques.
%\item Motivate why this problem is harder than the previously dealt problems.
%\item Motivate the advantage of tree-based algorithm.
%\item Propose tree-based algorithm for efficiently finding the best-match.
%\item Motivate dual-tree algorithms and propose the algorithm.
%\item Propose a new kind of tree -- cone trees
%\item Evaluate against linear scan on a variety of datasets.
%\end{citemize}
%%
%%
%% Question: Should we itemize this section for better clarity???
In this paper, we consider the general problem of efficient maximum inner-product search and propose two tree-based branch-and-bound algorithms along with a new data structure to solve this problem more efficiently that the naive linear scan. In the following section, we contrast this problem to the usual problems of nearest-neighbor search in metric spaces and best-matches with respect to cosine-similarity. This presents the need for explicit attention to this problem (equation \ref{eq:maxip}). However, we do motivate the use of the existing tree data structures for solving this task efficiently. In section \ref{sec:ref_ball}, we propose a simple branch-and-bound algorithm using the existing ball-tree data structure \cite{omohundro1989five} and a novel bound. In the following section (section \ref{sec:dual_tree}), we address the situation where there are multiple queries on the same set of points and propose a dual-tree branch-and-bound algorithm along with a new tree data structure, {\em cone trees}, to index the queries. The proposed algorithms are evaluated for their efficiency over a variety of data sets in section \ref{sec:expts}. Section \ref{sec:max_kernel} demonstrates how the proposed algorithms can be applied to the max-kernel operation with a general kernel function where it is not required to have an explicit representation of the points in the $\varphi$-space. In the final section (section \ref{sec:conclusions}), we provide our conclusions along with possible future directions for this work.  
%\vspace{-0.1in}
\section{Maximum Inner-product Search} \label{sec:related}
%\vspace{-0.05in}
%Explain why these two are different things and why MaxIP search is way harder.
%
%\begin{citemize}
%\item Max-IP search reduces to \NNS~only under certain conditions -- $\forall\ p \in S, ||p|| = \mbox{ some fixed constant}$. That is all points are normalized to have the same norm.
%\item MaxIP Search different from \NN~and maximum cosine similarity (because here you do not care about the norms) (need a figure to show this)
%\item MaxIP as a similarity function cannot be used to admit a locality sensitive hash function family since the distance function $1 - \mbox{sim}(x,y)$ does not satisfy the triangle inequality (\cite{charikar2002similarity}, Lemma 1).  
%\end{citemize}
%
The inner-product between two vectors is very closely related to the Euclidean distance between the points represented by these vectors as well as to the cosine-similarity between these to vectors. Numerous techniques exists for nearest-neighbor search in Euclidean metric space (see surveys like \cite{clarkson2006nearest}). Large scale best matching algorithms have also been developed for the cosine-similarity measure \cite{bayardo2007scaling}, with a lot of focus on text data. The problem of nearest-neighbor search (in metric space) has also been solved approximately with the widely popular {\em Locality-sensitive hashing} (LSH) method \cite{gionis1999similarity,indyk1998approximate}. The LSH technique has also been extended to other forms of similarity functions (as opposed to the distance as a dissimilarity function) like the cosine similarity \cite{charikar2002similarity}\footnote{An important thing to note here is that the similarity function used in Charikar et.al.\cite{charikar2002similarity} is not exactly the cosine-similarity. The distance between two points $p$ and $q$ was measured by $\theta / \pi$, where $\theta$ is the angle made the two points at the origin, making the similarity function $\left( 1 - \frac{\theta}{\pi} \right)$. This similarity function has a direct correspondence to the cosine similarity.}. The approximate max-kernel operations can also be solved efficiently with LSH under certain conditions on the kernel function. Some other techniques like dimension reduction \cite{rahimi2007random} and dual-tree algorithms \cite{klaas2005fast} have also been used to solve the approximate max-kernel operation efficiently. 

\begin{figure}[ht]
\centering
\includegraphics[width=0.6\columnwidth,clip=true,trim= 1.0in 3.7in 1.0in 1.7in]{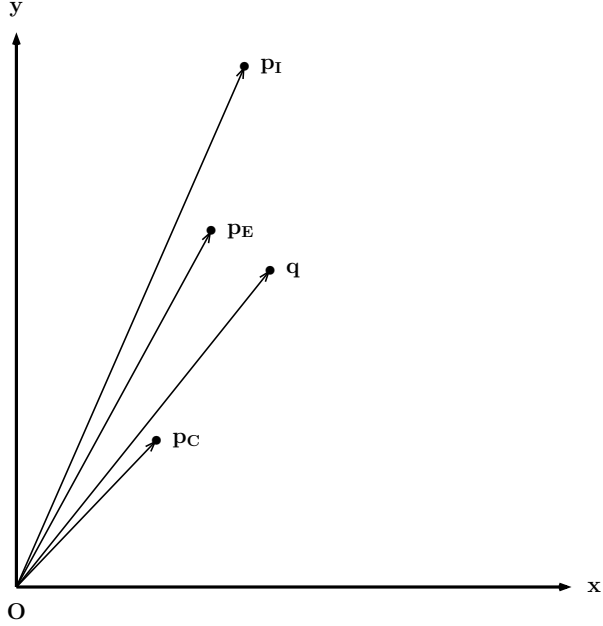}
\label{fig:best_matches}
%\vspace{-0.2in}
\caption{\textbf{Best matches:} For a given query $q$, $p_C$, $p_E$ and $p_I$ denote the best match with respect to the Cosine-similarity, the Euclidean distance and the Inner-product respectively. It is apparent from this figure that even on a plane, the best match with respect to these similarity functions can be very different.}
%\vspace{-0.1in}
\end{figure}
%\vspace{-0.1in}
\subsection{How is maximum inner-product search different from existing problems?}
In what follows, we will try to show that the problem stated in equation \ref{eq:maxip} is different from these existing problems. Hence techniques applied to these problems (like LSH) cannot be directly applied to this problem.
%\vspace{-0.1in}
\paragraph{Nearest-neighbor Search in Euclidean Space.}
The problem of finding the nearest-neighbor in this setting can be posed as finding a point $p\in S$ for a query $q$ such that:
\begin{eqnarray}
p  &  = & \arg \min_{r\in S} \twonormsq{q-r} 
% & = & \arg \min_{r\in S}\left( \twonormsq{q} + \twonormsq{r} - 2 \ip{q}{r} \right) \\
%& =  & \arg \min_{r \in S} \left( \twonormsq{r} - 2 \ip{q}{r} \right) \\
= \arg \max_{r \in S} \left( \ip{q}{r} - \frac{\twonormsq{r}}{2} \right ) \nonumber \\
& \not=  & \arg \max_{r\in S} \ip{q}{r} \mbox{ (unless $\twonormsq{r} = k\ \forall\ r \in S$)}. \nonumber
\end{eqnarray}
Hence, if the norms of all the points in $S$ are normalized to have the same length, then the problem of finding the best match with respect to the inner-product is equivalent to the problem of finding the nearest-neighbor in Euclidean metric space. However, without this restriction, the two problems can have potentially very different answers (figure \ref{fig:best_matches}). 
%\vspace{-0.1in}
\paragraph{Best-matching with Cosine-similarity.} The problem of finding the best match with respect to the cosine-similarity can be posed as finding a point $p \in S$ for a query $q$ such that 
\begin{eqnarray}
p & = & \arg \max_{r \in S} \frac{\ip{q}{r}}{\norm{q}\norm{r}} 
   =  \arg \max_{r \in S} \frac{\ip{q}{r}}{\norm{r}}  \nonumber \\
 & \not=  & \arg \max_{r\in S} \ip{q}{r} \mbox{ (unless $\norm{r} = k\ \forall\ r \in S$)}. \nonumber
\end{eqnarray}
As in the previous case, the best match with cosine similarity is the best match with inner-products if all the points in the set $S$ are normalized to have the same length. Under general conditions, the best matches with these two similarity functions can be very different (see figure \ref{fig:best_matches}). 
%\vspace{-0.1in}
\paragraph{Locality-sensitive Hashing.}
LSH has been applied to a wide variety of similarity functions. LSH involves constructing hashing functions such that each hash function $h$ must satisfy the following for any pair of points $r,p\in S$: 
\begin{equation}
\label{eq:lsh}
\Pr[h(r) = h(p)] = \mbox{sim}(r,p),
\end{equation}
where $\mbox{sim}(r,p) \in [0,1]$ is the similarity function of interest. For our situation, we can normalize our data set such that $\forall \ r\in S, \norm{r} \leq 1$\footnote{This normalization is different than the normalization mentioned before where all the points were normalized to have the same length. Here the lengths are normalized to be less than equal to one, but not equal to each other.}, and assume that the all the data is in the first quadrant (so that none of the inner-products go below zero). In that case, $\mbox{sim}(r,p) = \ip{r}{p} \in [0,1]$ is a valid similarity function of interest. 

It is known that for any similarity function to admit a locality sensitive hash function family (as defined in equation \ref{eq:lsh}), the distance function $\dist(r,p) = 1 - \mbox{sim}(r,p)$ must satisfy the triangle inequality (Lemma 1 in \cite{charikar2002similarity}). However, the distance function $\dist(r,p) = 1 - \ip{r}{p}$ does not satisfy the triangle inequality (even when all the points are restricted to the first quadrant)\footnote{Consider the following counter example: Let $x,y,z \in S$ be points such that $\norm{x} = \norm{y} = \norm{z} = 1$, and angles made between $x\ \&\ y$, $y\ \&\ z$ and $z\ \&\ x$ at the origin are $\left(\frac{\pi}{4} - 0.1 \right)$, $\frac{\pi}{4}$ and $\left( \frac{\pi}{2} - 0.3\right)$ respectively. In this case the inequality, $\dist(x,y) + \dist(y,z) \geq \dist(z,x)$ does not hold for $d(\cdot,\cdot) = 1 - \ip{\cdot}{\cdot}$. $\dist(x,y) = 0.226,\ \dist(y,z) = 0.293\ \&\ \dist(z,x) = 0.704$.}. So LSH cannot be applied to the inner product similarity function even when we assume that all the data lies in the first quadrant (which is quite a restrictive assumption).
%\vspace{-0.1in}
\paragraph{Efficient Max-kernel Operation.}
There are different existing techniques of solving this problem efficiently. For kernel functions with very high (possibly infinite) dimensional explicit representations, Rahimi, et.al., 2007 \cite{rahimi2007random}, propose a technique to transform these high-dimensional representations into lower-dimensions while still approximately preserving the inner-product to improve scalability. However, the final search still involves a linear scan over the set of points for the maximum inner-product or a fast nearest-neighbor search under the assumption that finding the nearest-neighbor is equivalent to maximizing the inner-product. For translation invariant kernels\footnote{Kernel functions $\mathcal{K}(p,q)$ which are dependent only the (Euclidean) distance between the points $p$ and $q$ are considered translation invariant kernels. The Gaussian RBF kernel is such a translation invariant kernel function.}, a tree-based recursive algorithm has been shown to scale to large sets \cite{klaas2005fast}. However, it is not clear how this algorithm can be extended to the general class of kernels. LSH is widely used for image matching in computer vision \cite{kulis2009kernelized}, but only for kernel functions that admit a locality sensitive hashing function \cite{charikar2002similarity}.

Hence, none of the existing techniques can be directly applied to our problem (equation \ref{eq:maxip}) without introducing inaccurate results or limiting assumptions.

%\vspace{-0.1in}
\subsection{Why is maximum inner-product search possibly harder?}
%\vspace{-0.02in}
%\begin{citemize}
%%\item $D(A,A) \not= 0$ unlike distances and metrics
%\item Distance function gives you triangle inequality, cosine similarity doesn't but the angle as a distance does. 
%\end{citemize}
Unlike the distance functions in metric space, inner products do not induce any form of triangle inequality (even under some assumptions as mentioned in the previous section). Moreover, this lack of any induced triangle inequality causes the similarity function induced by the inner products to have no admissible family of locality sensitive hashing functions. And any modification to the similarity function to conform to widely used similarity functions (like Euclidean distance or Cosine-similarity) will create inaccurate results. 

Moreover, inner-products lack a very basic property of generally used similarity functions -- the self similarity is high (generally the highest). For example, the Euclidean distance of a point to itself is 0; the cosine-similarity of a point to itself is 1. The inner-product of a point $x$ to itself is $\norm{x}^2$, which may be high or low depending on the value of the $\norm{x}$. Moreover, there can possibly be many other points like $y$ in the set such that $\ip{y}{x} > \norm{x}^2$. 

Hence, without any assumptions, this problem of obtaining the best match with respect to the maximum inner product is inherently harder than the previously dealt similar problems. This is possibly the reason why there is no existing work for this problem without any restrictions on the domain (at least to the best of our knowledge).  

%\vspace{-0.1in}
\subsection{Are trees the answer?}
%\vspace{-0.02in}
%\begin{citemize}
%\item Bounding inner-product with trees (similar to bounding distances with trees) open up avenues to apply techniques used in nearest-neighbor search in metric spaces to the problem of finding the maximum inner product. 
%\item Trees do well in low-medium dimensions, not so good in high dimensions.
%\item However, trees give an opportunity to do controlled approximate search -- error constrained or time constrained. (with something like LSH, you can't do either -- (i) no way to bound error, the bounds are only theoretical and no way to guarantee error bounds at search time (ii) inherently not an incremental algorithm so no way of doing time-constrained search).
%\item Trees require single construction, then they are ready for multiple levels of approximation -- error constrained or time constrained (as opposed to pre-hashing at multiple levels of approximation and trying everything and hoping for the best). 
%\item Trees can be learned for better performance (think about this) \cite{cayton2007learning, li2011Learning}
%\end{citemize}
In this paper, we explore the tree data structure for indexing the points and a branch-and-bound algorithm specifically for the task of maximum inner-product search. Tree data structures have been widely used for the task of nearest-neighbor search \cite{friedman1977algorithm, beygelzimer2006cover, preparata1985computational}. And even though the task of nearest-neighbor search is slightly different from the task of maximum inner-product search, we believe that trees can still be useful for this task. 

Trees are known to be good indexing schemes in low to medium dimensions, while some new tree data structures have been developed for data in high dimensions with some low dimensional structure \cite{dasgupta2008random, beygelzimer2006cover}. A hierarchical representation of the data is useful. In this paper, we try to solve the problem of exact maximum inner-product search. The hierarchical tree data structure provides a very intuitive extension to solve the problem approximately to gain efficiency \cite{ciaccia2000pac, ram2009rank}. %The approximation can be value approximation or order approximation (this might need explaining).

Moreover, if the search is to be performed with strict constraints -- error constraints or time constraints, the tree-based branch-and-bound algorithms can be easily adapted for that purpose. This is because these branch-and-bound algorithms are incremental algorithms. This is not possible with something like LSH -- LSH provides theoretical error bounds, but there is no way of ensuring the error constraint during the search. Moreover, LSH is inherently not an incremental algorithm, and hence cannot be used in a limited time setting.

An important advantage of trees is that the trees require a single construction -- the branch-and-bound algorithm adapts for the different levels of approximate and/or time limitations. Hashing techniques require multiple hashes for different levels of approximation. The usual norm is to pre-hash for multiple values of approximation. Trees can also be constructed by learning from the data using techniques from machine learning \cite{cayton2007learning, li2011Learning} to provide better accuracy and efficiency.

This is why we use trees to solve the problem. Trees might not be the best possible way to solve this problem, but trees do bring a lot of advantages with them.
%\newpage
%\vspace{-0.1in}
\section{Tree-based Search}
\label{sec:ref_ball}
%\vspace{-0.05in}
%\textit{
%\begin{citemize}
%\item[*] Better organization: motivate ball-trees 
%\item[->] explain ball trees and construction (pseudo-code)
%\item[->] explain branch and bound algorithm (and the pseudo code) 
%\item[->] the math of the bound (theoremize)
%\end{citemize}
%}
Ball trees~\cite{preparata1985computational, omohundro1989five} are binary space-partitioning trees that have been widely used for the task of indexing data sets. Every node in the tree represents a set of points and each node is subsequently indexed with a center and a ball enclosing all the points in the node. The set of point at a node is divided into two disjoint sets which form the child nodes. This partitions the space into (possibly overlapping) hyper-spheres. The tree is built hierarchically and a node is declared to be a leaf node if it contains a set of points of size below a threshold value $N_0$. 
%The tree $T$ is built hierarchically and each node in the tree is defined by the mean of the data in that node $(T.\mbox{center})$ and the radius of the ball around the mean enclosing all the points in the node ($T.\mbox{radius}$). The tree has leaves of size at most $N_0$. The splitting and the recursive tree construction algorithm is presented in Algorithms \ref{alg:metricsplits} \& \ref{alg:metrictree}.
\begin{figure}[ht]
\centering
\includegraphics[width=0.6\columnwidth,clip=true,trim= 2.0in 7.2in 2.0in 1.3in]{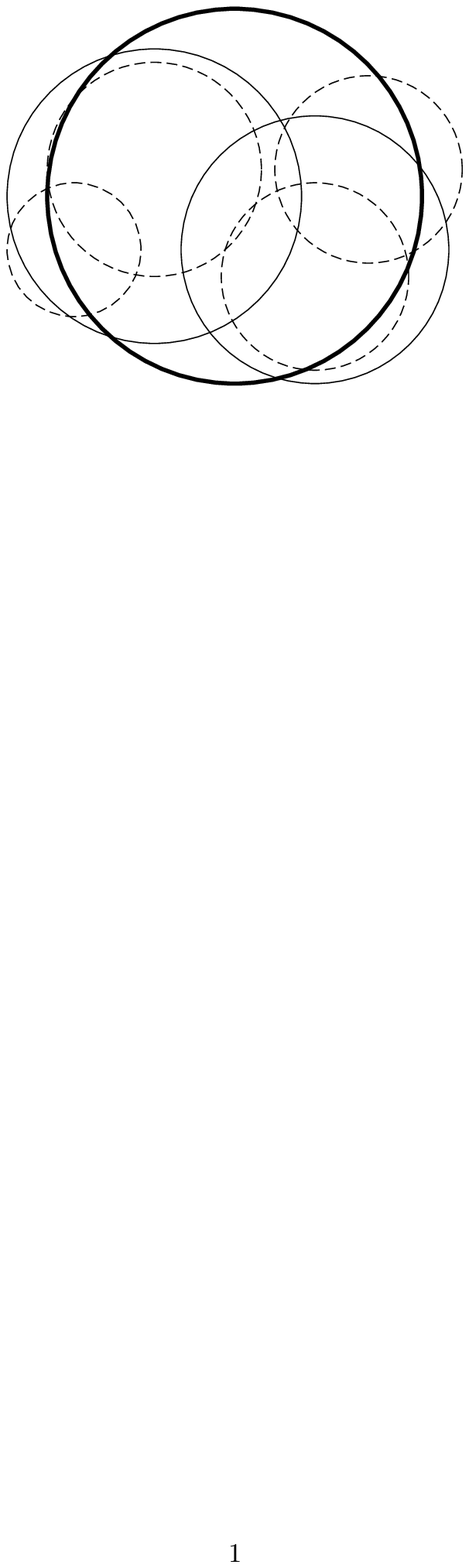}
\caption{\textbf{Ball-trees:} All the points are limited within the root ball (the bold-face circle). However, the subsequent balls does not necessarily lie within the parent ball -- the points still lie within the root ball, but the ball enclosing the points in the child node are not necessarily compact enough to lie within the parent ball. However, the child node would be confined within the parent node if we used hyper-rectangles instead of balls to index the data.}
\label{fig:ball-tree}
\end{figure}
\subsection{Tree Construction}
We use a simple ball tree construction heuristic that approximately picks a pair of pivot points which are farthest apart from each other \cite{omohundro1989five}, and splits the data by assigning the points to their closest pivot. The intuition behind this heuristic is that these two points might lie in the principal direction. The splitting and the recursive tree construction algorithm is presented in Algorithms \ref{alg:ball-splits} \& \ref{alg:ball-tree-construction} for completeness. 

The tree is very space efficient since every node only stores the indices of the item vectors instead of the item vectors themselves. Hence the matrix for the items is never duplicated. Another implementation optimization is that the vectors in the items' matrix are sorted in place (during the tree construction) such that all the items in the same leaf node are arranged serially in the matrix. This is to avoid any random access to the memory when accessing items in the same leaf node.
\begin{figure}[!htb]
\centering
 \begin{algorithm}[H]
{\small
 \caption{MakeBallTreeSplit(Data $S$)}
 \begin{algorithmic}
 \label{alg:ball-splits}
 \STATE Pick a random point $\xvec \in S$
 \STATE $A \leftarrow \arg\max_{\xvec' \in S} \twonormsq{\xvec - \xvec'}$
 \STATE $B \leftarrow \arg\max_{\xvec' \in S} \twonormsq{A - \xvec'}$
% \STATE $\wvec \leftarrow (B - A)$
% \STATE $b \leftarrow -\frac{1}{2} \left( \twonormsq{B} - \twonormsq{A} \right)$
% \STATE return $(\wvec,b)$
 \STATE return $(A,B)$
 \end{algorithmic}
}
\end{algorithm}
 \begin{algorithm}[H]
{\small
 \caption{MakeBallTree(Set of items $S$)}
 \begin{algorithmic}
 \label{alg:ball-tree-construction}
 \STATE Input -- Set $S$
 \STATE Output -- Tree $T$
 \STATE $T.S \leftarrow S$
 \STATE $T.\mu \leftarrow \mbox{mean}(S)$
 \STATE $T.R \leftarrow \max_{p \in S} \twonormsq{p - T.\mu}$
 \IF { $|S| \leq N_0$ }
    \STATE \textit{// Leaf node}
    \STATE return $T$
 \ELSE{}
    \STATE \textit{// else split the set}
    \STATE $(A,B) \leftarrow \mbox{MakeBallTreeSplit}(S)$
    \STATE $S_l \leftarrow \{p \in S \colon \twonormsq{p - A} \leq \twonormsq{p - B}\}$
    \STATE $S_r \leftarrow S \setminus S_l$
    \STATE $T.\mbox{lc} \leftarrow \mbox{MakeBallTree}(S_l)$
    \STATE $T.\mbox{rc} \leftarrow \mbox{MakeBallTree}(S_r)$
    \STATE return $T$
 \ENDIF
 \end{algorithmic}
}
\end{algorithm}
\caption{\textbf{Ball-tree Construction:} The object $T.S$ denotes the set of points in the node $T$. $T.\mu$ denotes the Euclidean mean of the items in the node $T$ and $T.R$ denotes the minimum radius of the ball centered around $T.\mu$ enclosing all the points in the node $T$. $T.\mbox{lc}$ and $T.\mbox{rc}$ denotes the left and right child of the tree node $T$.}
\label{fig:ball-tree-construction}
\end{figure}
\subsection{Branch-and-bound algorithm}
Ball trees are widely used for the task of nearest neighbor search and are known to be fairly scalable to moderately high dimensions \cite{omohundro1989five,liu2004investigation}.
The search usually employs the depth-first branch-and-bound algorithm. A nearest neighbor query is answered by traversing the tree in a depth-first manner by first going down the node closer to the query and bounding the minimum possible distance to the other branch with the triangle-inequality. If this bound is greater than the current neighbor candidate for the query, the branch is removed from computation.

An analogous greedy depth-first algorithm can be used for maximum inner-product search. But instead of traversing down the node closer to the query, the choice is made on the basis of the maximum possible inner-product between the query and any potential point from the node. The recursive depth-first branch and bound algorithm is presented in Algorithm \ref{alg:rec-single-search}. The search algorithm for a query ($q$) begins at the root of the tree (Alg. \ref{alg:single-tree-search}). At each step, the algorithm is at a tree node ($T$). It checks if the maximum possible inner-product between the query and any point in the node, $\mathbf{MIP}(q,T)$,  is any better than the current best-match for the query ($q.\mbox{bm}$). If the check fails, this branch of the tree is not explored any more. Otherwise, the algorithm recursively traverses the tree, exploring the branch with the better potential candidates in a depth-first manner. If the node is a leaf, the algorithm just finds the best-match within the leaf with a linear search (Alg. \ref{alg:linear-search}). This algorithm ensures that the exact solution (i.e., the best-match) is returned by the end of the algorithm.
%\subsubsection{The Algorithm}
%Using this upper bound (\ref{eq:ball_bound}) for the maximum possible inner product, we present the depth-first branch-and-bound algorithm to search for the $K$-highest inner products in Algorithm \ref{alg:metricsearch}. The algorithm begins at the root of the tree of items. At each subsequent step, the algorithm is at a tree node. Using the bound in equation \ref{eq:ball_bound}, the algorithm checks if the best possible item in this node is any better than the current best candidates for the user. If the check fails, this branch of the tree is not explored any more. Otherwise, the algorithm recursively traverses the tree, exploring the branch with the better potential candidates in a depth-first manner. If the node is a leaf, the algorithm just finds the best candidates within the leaf with the simple naive search. This algorithm ensures that the exact solution (i.e., the best candidates) is returned by the end of the algorithm.
%
%We employ a similar branch-and-bound algorithm for the purposes of searching for the $K$-highest inner products (as opposed to the minimum pairwise distance in $K$-nearest-neighbor search). 
\begin{figure}[!htb]
\centering
 \begin{algorithm}[H]
 \caption{LinearSearch(Query $q$, Reference Set $S$)}
 \label{alg:linear-search}
{\small
 \begin{algorithmic}
 \FOR{ each $p\in S$ }
    \IF{ $\ip{q}{p} > q.\lambda$ }
       \STATE $q.\mbox{bm} \leftarrow p$
       \STATE $q.\lambda \leftarrow \ip{q}{p}$
    \ENDIF
 \ENDFOR
 \end{algorithmic}
}
 \end{algorithm}
 \begin{algorithm}[H]
%{\small
 \caption{TreeSearch(Query $q$, Tree Node $T$)}
{\small
 \begin{algorithmic}
 \label{alg:rec-single-search}
% \IF {$q.\mbox{ub} <  q^T T.\mbox{center} + T.\mbox{radius} \cdot ||q||$}
 \IF {$q.\lambda <  \mathbf{MIP}(q, T)$}
 \STATE \textit{// This node has potential}
 \IF {\textit{isLeaf}(T)}
    \STATE LinearSearch($q$, $T.S$)
 \ELSE {}
    \STATE \textit{// best depth first traversal}
%    \STATE $I_l \leftarrow  q^T T.\mbox{left}.\mbox{center} $; $I_r \leftarrow q^T T.\mbox{right}.\mbox{center}  $;
    \STATE $I_l \leftarrow  \mathbf{MIP}(q, T.\mbox{lc})$; $I_r \leftarrow \mathbf{MIP}(q, T.\mbox{rc})$;
    \IF {$I_l \leq I_r$}
        \STATE TreeSearch$(q, T.\mbox{rc})$; TreeSearch$(q, T.\mbox{lc})$;
    \ELSE {}
        \STATE TreeSearch$(q, T.\mbox{lc})$; TreeSearch$(q, T.\mbox{rc})$;
    \ENDIF
 \ENDIF
 \ENDIF
 \STATE \textit{// Else the node is pruned from computation}
 \STATE return;
 \end{algorithmic}
}
\end{algorithm}
 \begin{algorithm}[H]
{\small
 \caption{FindExactMaxIP(Query set $V$, Reference Set $S$)}
 \begin{algorithmic}
 \label{alg:single-tree-search}
 \STATE $T \leftarrow $ MakeBallTree($S$)
 \FOR{ each $q \in V$ }
   \STATE $q.\lambda \leftarrow -\infty$;
   \STATE $q.\mbox{bm} \leftarrow \emptyset$;
   \STATE TreeSearch($q$, $T$);
   \STATE return $q.\mbox{bm}$;
 \ENDFOR
 \end{algorithmic}
}
\end{algorithm}
\caption{\textbf{Single-tree Search}: The object $q.\mbox{bm}$ contains the current best-match candidate for the query and $q.\lambda$ denotes the inner-product between the query $q$ and its current best-match $q.\mbox{bm}$. The function `$\mathbf{MIP}(q,T) = \ip{q}{T.\mu} + \norm{q} T.R$' denotes the upper bound on the maximum possible inner-product between the query $q$ and any point lying in the tree node $T$.}
\label{fig:single-ball-tree-search}
\end{figure}
\subsubsection{Bounding maximum inner-product with a ball}
Since the triangle inequality does not hold for the inner product, we present an novel analytical upper bound for the maximum possible inner product of a given point (in this case, the query $q$) with points in a ball. It is important to note that the information about the ball is limited to its center and its radius. For the rest of this section, we use the notation $\norm{\cdot}$ to denote the $\twonorm{\cdot}$. 

\begin{figure}[t]
\centering
\includegraphics[width=0.6\columnwidth,clip=true,trim= 1in 3.7in 1in 1.7in]{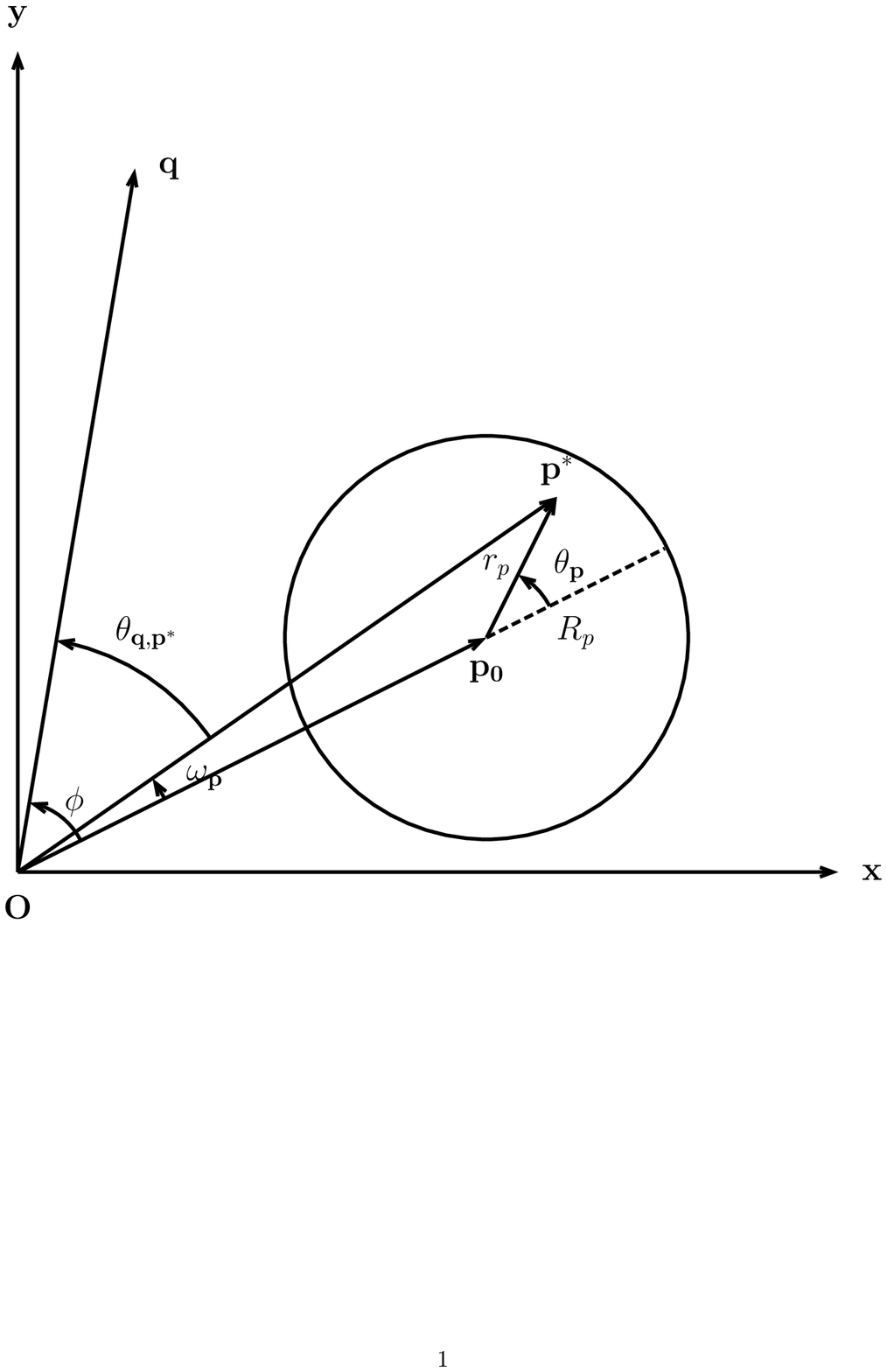}
\caption{Bounding with a ball}
\label{fig:ball_bound}
\end{figure}
\begin{theorem}
\label{thm:single-ball-bnd} 
Given a ball $\ball{p_0}{R_p}$ of points centered at $p_0$ with radius $R_p$ and (query) point $q$, the maximum possible inner product between the point $q$ and the ball $\ball{p_0}{R_p}$ is bounded from above by:
\begin{equation} \label{eq:mip_ball_bnd}
\max_{p \in \ball{p_0}{R_p}} \ip{q}{p} \leq \ip{q}{p_0} + R_p \norm{q}.
\end{equation}
\end{theorem}
\begin{proof}
Suppose that $p^*$ is the best possible match in the ball $\ball{p_0}{R_p}$ for the query $q$ and $r_p$ be the Euclidean distance between the ball center $p_0$ and $p^*$ (by definition, $r_p \leq R_p$). Let $\theta_p$ be the angle between the vector $\vec{p_0}$ and the vector $\vec{p_0 p^*}$, $\phi$ and $\omega_p$ be the angles made at the origin between the vector $\vec{p_0}$ and vectors $\vec{q}$ and $\vec{p^*}$ respectively (see figure \ref{fig:ball_bound}). The length of $p^*$ in terms of $p_0$ and $\theta_p$  is:
\begin{equation}
\label{eq:best_p_norm}
\norm{p^*} = \sqrt{ (\norm{p_0} + r_p \cos \theta_p)^2 + (r_p \sin \theta_p)^2 }.
\end{equation}
The angle $\omega_p$ can be expressed in terms of $p_0$ and $\theta_p$ as: 
\begin{equation}
\label{eq:best_p_angles}
\cos \omega_p = \frac{\norm{p_0} + r_p \cos \theta_p}{\norm{p^*}}, \sin \omega_p = \frac{r_p \sin \theta_p}{\norm{p^*}}.
\end{equation}
Let $\theta_{q,p^*}$ be the angle between the vectors $\vec{q}$ and $\vec{p^*}$. With the triangle inequality of angles, we have: %and the fact that $\cos(\theta) = \cos(-\theta)$, we have
%assumption that 
%This gives the following inequality regarding the angle between the query and the best possible angle (we assume that the angles lie in the range of $[-\pi, +\pi]$ instead of the usual $[0, 2\pi]$) :
$$|\theta_{q,p^*}| \geq |\phi - \omega_p|.$$
Assuming that the angles lie in the range $[-\pi,\pi]$ (instead of the usual $[0, 2\pi]$), and the fact that $\cos(\theta) = \cos(-\theta)$, we get:
\begin{equation}
\label{eq:m_angle_bounds}
\cos \theta_{q,p^*} \leq \cos (\phi - \omega_p),
\end{equation}
since $\cos (\cdot)$ is monotonically decreasing in the range $[0,\pi]$. Using this inequality we obtain the following bound for the highest possible inner-product between $q$ and any point in the ball:
\begin{eqnarray*}
\max_{p \in \ball{p_0}{R_p}} \ip{q}{p} & =  & \ip{q}{p^*} \mbox{(by assumption)} \\
 & = & \norm{q}\norm{p^*} \cos \theta_{q, p^*} \\
 & \leq & \norm{q}\norm{p^*} \cos (\phi - \omega_p),
\end{eqnarray*}
where the last inequality follows from equation \ref{eq:m_angle_bounds}. Substituting equations \ref{eq:best_p_norm} \& \ref{eq:best_p_angles} in the above inequality, we have
\begin{eqnarray*}
\max_{p \in \ball{p_0}{R_p}} \ip{q}{p} & \leq & \norm{q} \left( \cos \phi (\norm{p_0} + r_p \cos \theta_p)  + \sin \phi (r_p \sin \theta_p) \right) \\
 & \leq & \norm{q} \max_{\theta_p} \left( \cos \phi (\norm{p_0} + r_p \cos \theta_p) + \sin \phi (r_p \sin \theta_p) \right) \\
 & = & ||q|| \left( \cos \phi (||p_0|| + r_p \cos \phi)  + \sin \phi (r_p \sin \phi) \right) \\
 & \leq & \norm{q} \left( \cos \phi (\norm{p_0} + R_p \cos \phi) + \sin \phi (R_p \sin \phi) \right)\ \ (\mbox{since } r_p \leq R_p).\\
\end{eqnarray*}
The second inequality comes from the definition of maximum. The following equality comes from maximizing over $\theta_p$. This gives us the optimal value of $\theta_p = \phi$. Simplifying the final inequality gives us equation \ref{eq:mip_ball_bnd}.
%\begin{equation}
%\label{eq:ball_bound}
%\max_{p \in \ball{p_0}{R_p}}  \ip{q}{p}  \leq \ip{q}{p_0} + R_p ||q||.
%\end{equation}
\end{proof}
For the tree-search algorithm (Alg. \ref{alg:rec-single-search}), we set the maximum possible inner-product between $q$ and a tree node $T$ as $$\textbf{MIP}(q,T) = \ip{q}{T.\mu} + T.R \norm{q}.$$ This upper bound can be computed in almost the same time required for a single inner-product (since the norms of the queries can be pre-computed before searching the tree). This algorithm is evaluated against the naive linear search algorithm in section \ref{sec:expts}.

%\newpage
\section{Dual-tree based Search}
%\textit{
%\begin{citemize}
%%\item[*] Better organization: dual-tree motivation 
%%\item[->] dual-tree branch and bound algorithm 
%\item[->] math for bound with ball trees (fix proof)
%%\item[->] motivate cone trees 
%%\item[->] cone tree construction 
%\item[->] math for cone tree bound. (fix proof)
%\end{citemize}
%}
\label{sec:dual_tree}
For a set of queries, the tree can be traversed separately for each query. However, if the set of queries is very large, a common technique to improve efficiency of querying is to index the queries in the form of a tree as well. The search is then subsequently done by traversing both trees simultaneously using the {\em dual-tree} algorithm \cite{gray2000nbody}. The basic idea is to amortize the cost of tree-traversal for a set of queries which are very similar to each other and would follow (approximately) the same path down the tree.  The dual-tree algorithms have been applied to different tree-based algorithms like nearest-neighbor search \cite{gray2000nbody} and kernel density estimation \cite{gray2003nonparametric} with provable theoretical runtime bounds \cite{ram2009linear}. 

\subsection{Dual-tree Branch-and-bound Algorithm}
The generic dual-tree algorithm is presented in Algorithm~\ref{alg:dual-rec-search}. Similar to the Algorithm \ref{alg:rec-single-search}, the algorithm traverses down the tree on the reference set $S$ (referred to as the {\em RTree}). However, the algorithm also traverses down the tree on the set $V$ of queries ({\em QTree}), resulting in a four-way recursion. At each step, the algorithm is at a QTree node $Q$ and a RTree node T. For every $Q$, the value $Q.\lambda$ denotes the minimum inner-product between any query in $Q$ and its current best-match candidate. If this value is greater than the maximum possible inner product, $\mathbf{MIP}(Q,T)$, between any query in $Q$ and any reference point in $T$, this part of the recursion is no longer explored. When the algorithm is at the leaf level of both the trees, it obtains the best-matches for each query in the QTree leaf by doing a linear scan over the RTree leaf. 

%\subsection{Dual-tree search algorithm}
\begin{figure}[!htb]
\centering
 \begin{algorithm}[H]
{\small
 \caption{DualSearch(QTree Node $Q$, RTree Node $T$)}
 \begin{algorithmic}
 \label{alg:dual-rec-search}
 \IF {$Q.\lambda <  \mathbf{MIP}(Q,T)$}
  \STATE \textit{// This node has potential}
  \IF {\textit{isLeaf}($T$) \& \textit{isLeaf}($Q$)}
    \FOR {each  $q \in Q.S$ }
       \STATE LinearSearch($q, T.S$)
    \ENDFOR
    \STATE $Q.\lambda \leftarrow \min_{q\in Q.S} q.\lambda$
  \ELSIF {\textit{isLeaf}($T$)}
    \STATE DualSearch$(Q.\mbox{lc}, T)$; DualSearch$(Q.\mbox{rc}, T)$;
    \STATE $Q.\lambda \leftarrow \min\{Q.\mbox{lc}.\lambda, Q.\mbox{rc}.\lambda \}$
  \ELSIF {\textit{isLeaf}($Q$)}
    \STATE $I_l \leftarrow  \mathbf{MIP}(Q, T.\mbox{lc}) $; $I_r \leftarrow \mathbf{MIP}(Q, T.\mbox{rc}) $;
    \IF {$I_l \leq I_r$}
      \STATE DualSearch$(Q, T.\mbox{rc})$; DualSearch$(Q, T.\mbox{lc})$;
    \ELSE {}
      \STATE DualSearch$(Q, T.\mbox{lc})$; DualSearch$(Q, T.\mbox{rc})$;
    \ENDIF
 \ELSE{}
    \STATE \textit{// best depth first traversal}
    \STATE $I_l \leftarrow  \mathbf{MIP}(Q.\mbox{lc}, T.\mbox{lc}) $; $I_r \leftarrow \mathbf{MIP}(Q.\mbox{lc}, T.\mbox{rc}) $;
    \IF {$I_l \leq I_r$}
      \STATE DualSearch$(Q.\mbox{lc}, T.\mbox{rc})$; DualSearch$(Q.\mbox{lc}, T.\mbox{lc})$;
    \ELSE {}
      \STATE DualSearch$(Q.\mbox{lc}, T.\mbox{lc})$; DualSearch$(Q.\mbox{lc}, T.\mbox{rc})$;
    \ENDIF
    \STATE $I_l \leftarrow  \mathbf{MIP}(Q.\mbox{rc}, T.\mbox{lc}) $; $I_r \leftarrow \mathbf{MIP}(Q.\mbox{rc}, T.\mbox{rc}) $;
    \IF {$I_l \leq I_r$}
      \STATE DualSearch$(Q.\mbox{rc}, T.\mbox{rc})$; DualSearch$(Q.\mbox{rc}, T.\mbox{lc})$;
    \ELSE {}
      \STATE DualSearch$(Q.\mbox{rc}, T.\mbox{lc})$; DualSearch$(Q.\mbox{rc}, T.\mbox{rc})$;
    \ENDIF
    \STATE $Q.\lambda \leftarrow \min\{Q.\mbox{lc}.\lambda, Q.\mbox{rc}.\lambda \}$
 \ENDIF
 \ENDIF
 \STATE \textit{// Else the node is pruned from computation}
% \STATE return;
 \end{algorithmic}
}
\end{algorithm}
\begin{algorithm}[H]
{\small
 \caption{FindExactMaxIPDualTree(Query Set $V$, Reference Set $S$)}
 \begin{algorithmic}
 \label{alg:dual-exact-search}
% \STATE \textit{FIX IT}
 \STATE $T \leftarrow $ MakeBallTree($S$)
 \STATE $Q \leftarrow $ MakeQueryTree($V$)
 \STATE $\forall$ trees nodes $Q'$ in the tree $Q$, $Q'.\lambda \leftarrow -\infty$;
 \STATE $\forall$ queries $q \in V$, $q.\mbox{bm} \leftarrow \emptyset$, $q.\lambda \leftarrow -\infty$;
 \STATE DualSearch($Q$, $T$);
 \STATE $\forall$ queries $q \in V$, return $q.\mbox{bm}$;
 \end{algorithmic}
}
\end{algorithm}
\caption{\textbf{Dual-tree Search}: The tree-building subroutine for the set of queries ``\textit{MakeQueryTree}'' can be the ``\textit{MakeBallTree}'' subroutine (Alg. \ref{alg:ball-tree-construction}) or the ``\textit{MakeConeTree}'' subroutine (Alg. \ref{alg:cone-tree-construction}). The object $q.\mbox{bm}$ contains the current best-match for the query $q$. $Q.\lambda$ denotes the lowest affinity between any query in the node $Q$ and its current best-match. The function $\mathbf{MIP}(Q,T)$ denotes the upper bound on the maximum possible inner-product between any query in the node $Q$ and any point in the node $T$.}% For queries indexed in a ball-tree, $\mathbf{MIP}(Q,T)$ can be defined by Equation \ref{eq:mip_ball_ball_bnd}. For queries indexed in a cone-tree, $\mathbf{MIP}(Q,T)$ is defined by Equation \ref{eq:mip_cone_ball_bnd}.}
%\label{fig:metric-tree-search}
\end{figure}

In this section, we explore two ways of indexing the queries -- (1) indexing the queries using the ball-tree (2) indexing the queries using a novel data structure, the {\em cone-tree}. In the following subsections, we derive expressions for $\mathbf{MIP}(Q,T)$ each of these kinds of QTree.
\subsection{Ball Tree for Queries}
\begin{figure}[t]
\centering
\includegraphics[width=0.6\columnwidth,clip=true,trim= 1in 3.7in 1in 1.7in]{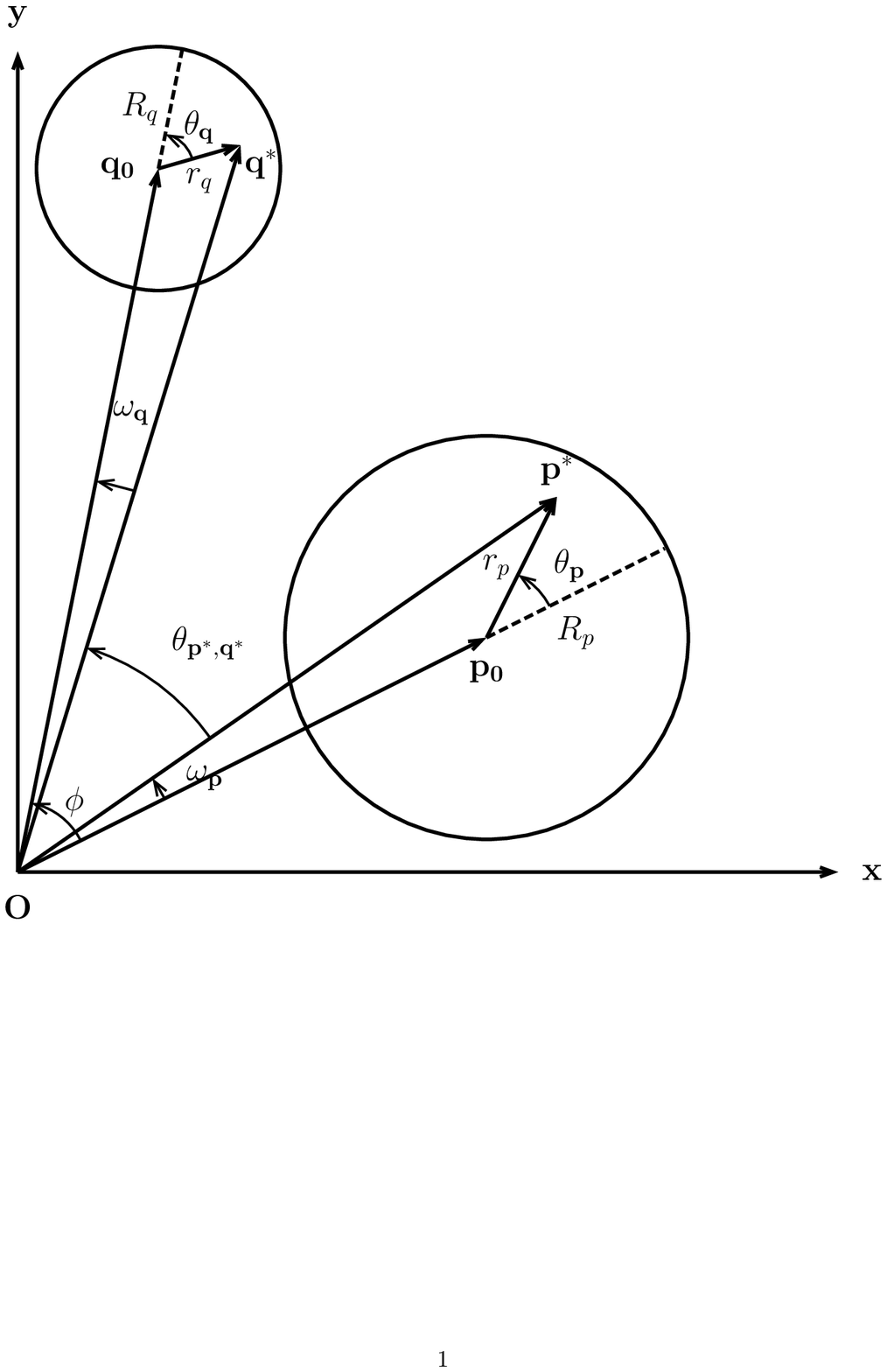}
\caption{Bounding between two balls.}
\label{fig:ball_ball_bound}
\end{figure}

\begin{theorem} \label{thm:ball-ball-bound}
Given two balls $\ball{p_0}{R_p}$ and $\ball{q_0}{R_q}$ centered at $p_0$ and $q_0$ with radius $R_p$ and $R_q$ respectively, the maximum possible inner-product with any pair of points $p \in \ball{p_0}{R_p}$ and $q \in \ball{q_0}{R_q}$ is bounded from above by: 
\begin{equation} \label{eq:mip_ball_ball_bnd}
\max_{p \in B_{p_0}^{R_p}, q \in B_{q_0}^{R_q}} \ip{q}{p} \leq  \ip{q_0}{p_0} + R_q R_p + \norm{q_0} R_p + \norm{p_0} R_q.
\end{equation}
\end{theorem}
\begin{proof}
Consider the pair of point $(p^*, q^*), p^* \in B_{p_0}^{R_p}, q^* \in B_{q_0}^{R_q}$ be such that 
\begin{equation}
\label{eq:max_ball_ball}
\ip{q^*}{p^*} = \max_{p \in B_{p_0}^{R_p}, q \in B_{q_0}^{R_q}} \ip{q}{p}.
\end{equation}

Let $\theta_p$ be the angle $\vec{p_0}$ makes with the vector $\vec{p_0 p^*}$, and $\theta_q$ be the corresponding angle in the query ball. Let $\omega_p$ be the angle between the vectors $\vec{p_0}$ and $\vec{p^*}$ and $\omega_q$ be the angle between the vectors $\vec{q_0}$ and $\vec{q^*}$. Let $r_p$ be the distance between $p_0$ and $p^*$, $r_q$ be the distance between $q_0$ and $q^*$. Finally, let $\phi$ be the angle made between $p_0$ and $q_0$ at the origin. 

Some facts for the ball $\ball{p_0}{R_p}$ (the facts are analogous for the ball $\ball{q_0}{R_q}$):
$$\norm{p^*} = \sqrt{\norm{p_0}^2 + r_p^2 + 2\norm{p_0} r_p \cos\theta_p},$$ 
$$ \cos \omega_p = \frac{\norm{p_0} + r_p \cos \theta_p}{\norm{p^*}}, \sin \omega_p = \frac{r_p \sin\theta_p}{\norm{p^*}}.$$
Using the triangle inequality of the angles, we know that:
\begin{equation*}
|\theta_{q^*, p^*}| \geq |\phi - (\omega_p + \omega_q)|, 
\end{equation*}
giving us the following:
\begin{equation}
\label{eq:max_omega} \ip{q^*}{p^*}  = \norm{p^*}\norm{q^*} \cos (\phi - (\omega_p+\omega_q)) \\
\end{equation}
Replacing $\omega_p$ and $\omega_q$ with $\theta_p$ and $\theta_q$ by using the aforementioned equalities (similar to the techniques in proof for theorem \ref{thm:single-ball-bnd}), we have:
\begin{eqnarray}
\label{eq:max_theta} \ip{q^*}{p^*} & = & \ip{q_0}{p_0} + r_p r_q \cos(\phi-(\theta_p + \theta_q)) + r_p\norm{q_0} \cos (\phi-\theta_p) + r_q \norm{p_0} \cos (\phi - \theta_q) \nonumber \\
\label{eq:max_max} & \leq & \max_{r_p, r_q, \theta_p, \theta_q} \ip{q_0}{p_0} + r_p r_q \cos(\phi-(\theta_p + \theta_q)) + r_p\norm{q_0} \cos (\phi-\theta_p) \nonumber + r_q \norm{p_0} \cos (\phi - \theta_q) \\
\label{eq:cosine_bnd} & \leq & \max_{r_p, r_q}  \ip{q_0}{p_0} + r_p r_q + r_q \norm{p_0} + r_p \norm{q_0} \ \ (\mbox{since $\cos(\cdot) \leq 1$}),  \\
\label{eq:ball_ball_bnd} & \leq & \ip{q_0}{p_0} + R_p R_q + R_q \norm{p_0} + R_p \norm{q_0},
\end{eqnarray}
where the first inequality comes from the definition of $\max$ and the final inequality comes from the fact that $r_p \leq R_p$, $r_q \leq R_q$.
\end{proof}
For the dual-tree search algorithm (Alg. \ref{alg:dual-rec-search}), the maximum-possible inner-product between two tree nodes $Q$ and $T$ is set as
\begin{equation*}
\mathbf{MIP}(Q,T) = \ip{q_0}{p_0} + R_p R_q + R_q \norm{p_0} + R_p \norm{q_0}.
\end{equation*}
It is interesting to note that this upper bound bound reduces to the bound in theorem \ref{thm:single-ball-bnd} when the ball containing the queries is reduced to a single point, implying $R_q = 0$. 

%\newpage
\begin{figure}[!thb]
\centering
\includegraphics[width=0.55\columnwidth,clip=true,trim= 1.0in 3.7in 1.0in 1.7in]{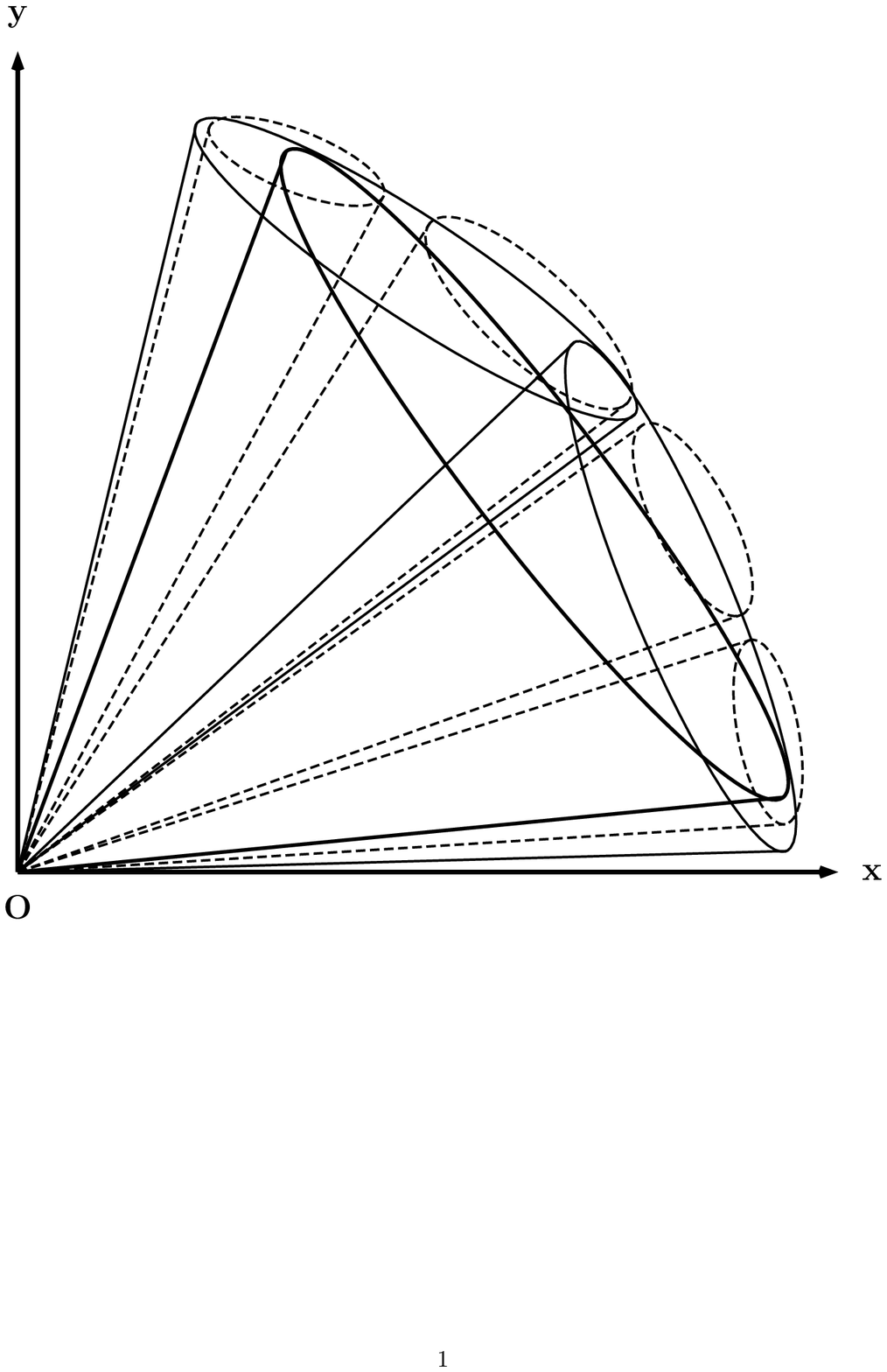}
\caption{\textbf{Cone-tree:} These cones are open cones and only the angle made at the origin with the axis of the cone is bounded for every point in the cone. The norms of the queries are not bounded at all.}
\label{fig:cone-tree}
\end{figure}
\subsection{Cone-trees for Queries}
%\label{sec:query_cone}
%\textit{
%\begin{itemize}
%\item Motivate why query-cone-ing is useful
%\item Cone tree building
%%\item Cone tree diagram
%\item Cone tree bounding
%\end{itemize}
%}
%
An interesting fact is that in equation \ref{eq:maxip}, the point $p$, where the maximum is achieved, is independent of the norm $||q||$ of the query $q$. Let $\theta_{q,r}$ be the angle between the $q$ and $r$ at the origin, then the task of searching for the maximum inner-product is equivalent to search for a point $p\in S$ such that:
\begin{equation}\label{eq:maxip_mod}
p = \arg\max_{r \in S} \norm{r} \cos \theta_{q,r}.
\end{equation}
This implies that we only care about the direction of the queries irrespective of their norms. For this reason, we propose the indexing of the queries on the basis of their direction (from the origin) to form a {\em cone-tree} (figure \ref{fig:cone-tree}). The queries are hierarchically indexed as (possibly overlapping) open cones. Each cone is represented by a vector, which corresponds to its axis, and an angle, which corresponds to the maximum angle made by any point within the cone with the axis at the origin.

\vspace{-0.1in}
\subsubsection{Cone-tree Construction}
\vspace{-0.05in}
The cone-tree construction is very similar to the ball-tree construction. The only difference is the use of cosine similarity instead of the Euclidean distances for the task of splitting. The cone-tree construction pseudo-code is presented in Figure \ref{fig:cone-tree}.

\begin{figure}[t]
\centering
 \begin{algorithm}[H]
{\small
 \caption{MakeConeTreeSplit(Data $Q$)}
 \begin{algorithmic}
 \label{alg:cone-splits}
 \STATE Pick a random point $\xvec \in Q$
 \STATE $A \leftarrow \arg\min_{\xvec' \in S} \cos \theta_{\xvec, \xvec'}$
 \STATE $B \leftarrow \arg\min_{\xvec' \in S} \cos \theta_{A, \xvec'}$
 \STATE return $(A,B)$.
% \STATE $\wvec \leftarrow (B - A)$
% \STATE $b \leftarrow -\frac{1}{2} \left( ||B||^2 - ||A||^2 \right)$
% \STATE return $(\wvec,b)$
 \end{algorithmic}
}
\end{algorithm}
 \begin{algorithm}[H]
{\small
 \caption{MakeConeTree(Set of items $S$)}
 \begin{algorithmic}
 \label{alg:cone-tree-construction}
 \STATE Input -- Set $S$
 \STATE Output -- Tree $T$
 \STATE $T.S \leftarrow S$
 \STATE $T.\mu \leftarrow \mbox{mean}(S)$
 \STATE $T.C \leftarrow \min_{p \in S} \cos \theta_{T.\mu,p}$
 \IF { $|S| \leq N_0$ }
%    \STATE \textit{// Leaf node}
    \STATE return $T$
 \ELSE{}
%    \STATE \textit{// else split the set}
    \STATE $(A,B) \leftarrow \mbox{MakeConeTreeSplit}(S)$
    \STATE $S_l \leftarrow \{p \in S \colon \cos \theta_{A,p} > \cos \theta_{B,p} \}$
    \STATE $S_r \leftarrow S \setminus S_l$
    \STATE $T.\mbox{lc} \leftarrow \mbox{MakeConeTree}(S_l)$
    \STATE $T.\mbox{rc} \leftarrow \mbox{MakeConeTree}(S_r)$
    \STATE return $T$
 \ENDIF
 \end{algorithmic}
}
\end{algorithm}
\caption{\textbf{Cone-tree Construction:} The object $T.S$ denotes the set of points in the node $T$, $T.\mu$ denotes the Euclidean mean of the items in the node $T$ and $T.C$ denotes the cosine of the maximum angle made by any point in the node with $T.\mu$ at the origin. The angle made between any two points $A$ and $B$ at the origin is denoted by $\theta_{A,B}$. }
\label{fig:cone-tree-construction}
\end{figure}
\subsubsection{Cone-Ball Bound}
Since the norms of the queries do not affect the solution in equation \ref{eq:maxip_mod}, we assume that the norms of the queries are all equal to 1 for convenience. 
\begin{figure}[t]
\centering
\includegraphics[width=0.6\columnwidth,clip=true,trim= 1in 3.7in 1in 1.7in]{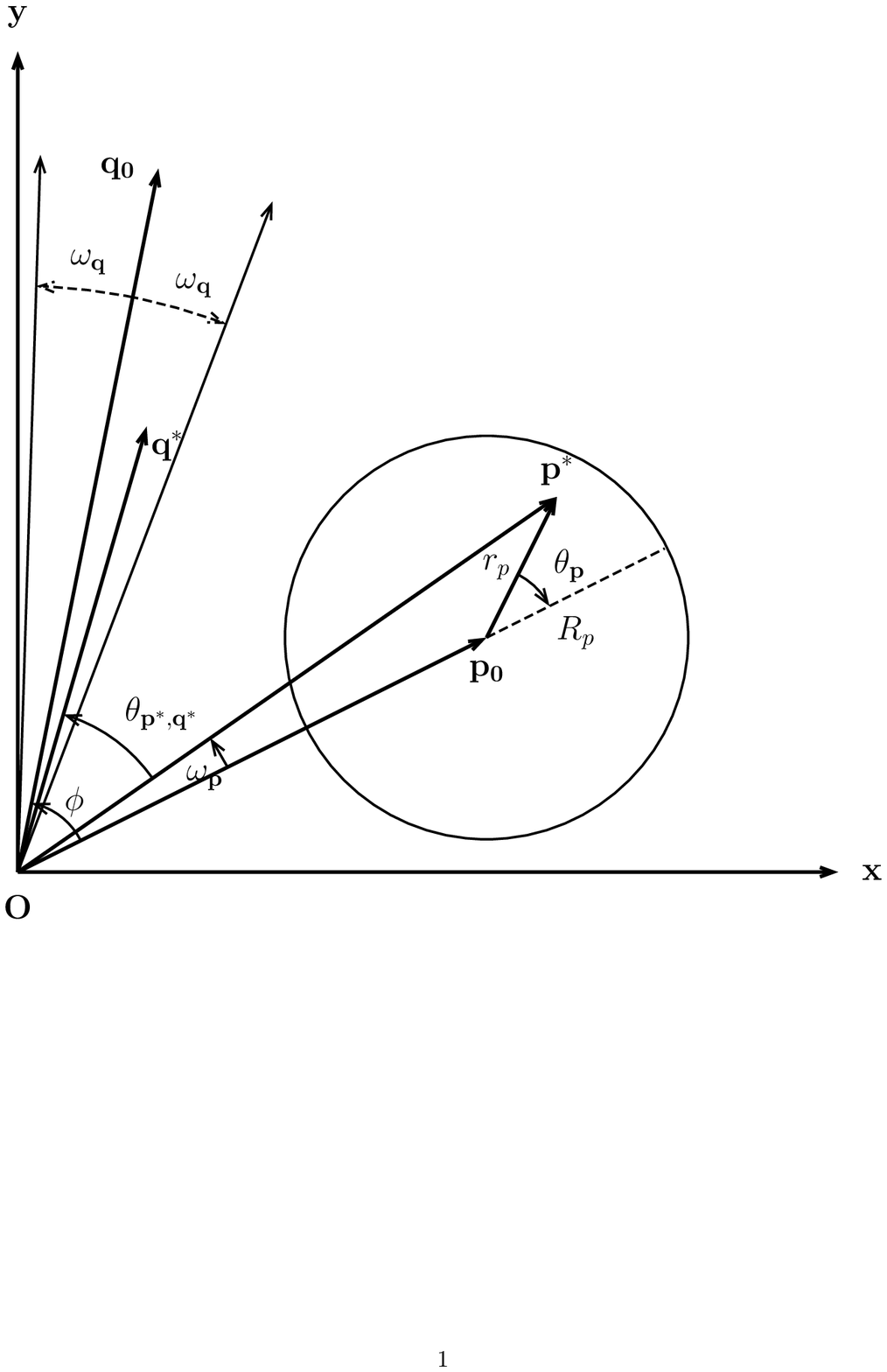}
\caption{Bounding between a ball and a cone}
\label{fig:ball_cone_bound}
\end{figure}

\begin{theorem}\label{thm:cone-ball-bnd}
Given a ball $\ball{p_0}{R_p}$ of points centered at $p_0$ with radius $R_p$ and a cone $\cone{q_0}{\omega_q}$ with the axis of the cone $q_0$ and aperture\footnote{The aperture of the cone is twice the angle made between the axis and the perimeter of the cone.} of $2\omega_q \geq 0$, the maximum possible inner-product between any pair of points $p \in \ball{p_0}{R_p}$, $q \in \cone{q_0}{\omega_q}$ is bounded from above by:
\begin{eqnarray}
\max_{q\in \cone{q_0}{\omega_q}, p\in \ball{p_0}{R_p}} \ip{q}{p} & = & \max_{q\in \cone{q_0}{\omega_q}, p\in \ball{p_0}{R_p}} \norm{p} \cos \theta_{q,p} \nonumber \\
\label{eq:mip_cone_ball_bnd}& \leq & \norm{p_0} \cos (\{|\phi| - \omega_q\}_{+}) + R_p,
\end{eqnarray}
where $\phi$ is the angle made between $p_0$ and $q_0$ at the origin and the function $\{x\}_+ = \max\{x, 0\}$.
\end{theorem}
\begin{proof}
%Then we have the following:  
%\begin{eqnarray}
%\label{eq:max_cone_ball} \max_{q\in C_{q_0}^{\omega_q}, p\in B_{p_0}^{R_p}} \ip{q}{p} & = & \max_{q\in C_{q_0}^{\omega_q}, p\in B_{p_0}^{R_p}} \norm{q}\norm{p} \cos \theta_{q,p}\\
%\label{eq:max_cone_ball_no_qnorm} & = & \max_{q\in C_{q_0}^{\omega_q}, p\in B_{p_0}^{R_p}} \norm{p} \cos \theta_{q,p} \\
%& & \ \ \ \mbox{ (assuming that $\norm{q} = 1\ \forall q$) } \nonumber
%\end{eqnarray}
%Here we assume that all the queries (and query centers) have been normalized to have unit norm since the norm of the queries do not affect the results.
There are two cases to consider here:
\begin{citemize}
\item[(i)] $|\phi| < \omega_q$
\item[(ii)] $|\phi| \geq \omega_q$
\end{citemize}
For case (i), the center $p_0$ of the ball $\ball{p_0}{R_p}$ lies within the cone $\cone{q_0}{\omega_q}$, implying that 
\begin{equation}
\max_{q\in \cone{q_0}{\omega_q}, p\in \ball{p_0}{R_p}} \norm{p} \cos \theta_{q,p} \leq \norm{p_0} + R_p.
\end{equation}
since there could be some query $q^* \in \cone{q_0}{\omega_q}$ which is in the same direction as $p_0$, giving the maximum possible inner-product.

For case (ii), let us assume that $\phi \geq 0$ without loss of generality. Then $\phi \geq \omega_q$. Continuing with the similar notation as in theorem \ref{thm:single-ball-bnd} \& \ref{thm:ball-ball-bound} for the best pair of points $(q^*, p^*)$ as well as the notation from figure \ref{fig:ball_cone_bound}, we can say that
\begin{equation}
|\theta_{p^*, q^*}| \geq |\phi - \omega_q - \omega_p|
\end{equation}
Since $\omega_q$ is fixed, we can say that 
\begin{eqnarray}
\max_{q\in C_{q_0}^{\omega_q}, p\in B_{p_0}^{R_p}} \norm{p} \cos \theta_{q,p}  & \leq &  \norm{p^*} \cos \theta_{q^*, p^*} \mbox{ (by def.)}\nonumber \\
& \leq & \norm{p^*} \cos(\phi - \omega_q - \omega_p).
\end{eqnarray}
Expressing $\norm{p^*}$ and $\omega_p$ in terms of $\norm{p_0}, r_p$ and $\theta_p$, and then subsequently maximizing over $\theta_p$ and using the fact that $r_p \leq R_p$, we get that
\begin{equation}
\label{eq:max_cone_ball_prefinal}
\max_{q\in C_{q_0}^{\omega_q}, p\in B_{p_0}^{R_p}} \norm{p} \cos \theta_{q,p} \leq \norm{p_0} \cos (\phi - \omega_q) + R_p.
\end{equation}
Combining case (i) and (ii), we obtain equation \ref{eq:mip_cone_ball_bnd}. 
%
%
%Using the same notation and technique as the single ball-tree bounding (replacing $\phi$ with $\phi - \omega_q$), we arrive to the following bound (arriving to the optimal value of $\theta_p = \phi - \omega_q$):
%\begin{equation}
%\label{eq:max_cone_ball_final}
%\max_{q\in C_{q_0}^{\omega_q}, p\in B_{p_0}^{R_p}} \norm{p} \cos \theta_{q,p} \leq \norm{p_0} \cos (\{\phi - \omega_q\}_{+}) + R_p.
%\end{equation}
\end{proof}

%\newpage
\section{Experiments and Results}
\label{sec:expts}
%\begin{citemize}
%\item[*] Explain that the only comparison is naive
%%\item[*] table of datasets with size?
%%\item[*] table for speedups with $K=1$.
%%\item[*] tree construction time
%%\item[*] Can add (projected) naive timings next to the construction times for contrast 
%\item[**] Speedup bars for $K=1,2,5,10$ for single, dual ball-ball and dual ball-cone.
%\end{citemize}

In this section, we evaluate the efficiency of the two proposed algorithms \ref{alg:single-tree-search} \& \ref{alg:dual-exact-search}. For the dual-tree algorithm, we use the two variations -- (i) the set of queries indexed as a ball-tree (referred to as Alg. \ref{alg:dual-exact-search}(B)), (ii) the set of queries indexed as a cone-tree (referred to as Alg. \ref{alg:dual-exact-search}(C)). Since we are not aware of any efficient exact method for maximum inner-product search, we compare our proposed algorithms to the linear search algorithm (Alg. \ref{alg:linear-search}). We report the speedup of the proposed algorithms over linear search. Speedup is defined as the ratio of the time taken by the linear search and the time taken by the evaluated algorithm. For the trees, the leaf size $N_0$ can be selected by cross-validation (choosing the leaf size giving the highest speedup). However, for our experiments, we choose a ad hoc value of $N_0 = 20$ for all datasets to demonstrate the gain in efficiency without any expensive cross-validation. 

\paragraph{Datasets.} We use a variety of datasets from different fields of data mining. We use the following collaborative filtering datasets: MovieLens \cite{GroupLens}, Netflix \cite{bennett2007netflix} and the Yahoo! Music \cite{dror2011yahoo} datasets. After the matrix factorization stage, the matrix of item-vectors is used as the reference set and the matrix of user-vectors is used as the set of queries. For text data, we use the LiveJournal blog moods data set \cite{Kim2011pre}. We also use the MNIST digits dataset \cite{lecun2000mnist} for evaluation. We also use three astronomy datasets -- LCDM \cite{lupton2001sdss}, PSF and SJ2. A synthetic data set (U-Rand) of uniformly random points in 20 dimensions is used to evaluate the performance of the tree-based algorithms on data sets without any underlying structure. The rest of the datasets are widely used machine learning data sets from the UCI machine learning repository \cite{ucimlrepository}. The details of the datasets are presented in Table \ref{tab:datasets} and the size of the datasets (in bytes) is presented in figure \ref{fig:data_total_plot}. For the collaborative filtering datasets, there is a clear definition of the reference set (the items) and the set of queries (the users). For the rest of the data sets, we randomly split the datasets into query and reference sets. 
\begin{table}[t]
\centering
\begin{small}
\begin{tabular}{|l|c|c|c|} \hline
\textsc{Dataset} & $\dims$ & $|S|$ & $|V|$ \\ \hline
Bio& 74& 210,409& 75,000 \\ \hline
Corel& 32& 27,749& 10,000 \\ \hline
Covertype& 55& 431,012& 150,000 \\ \hline
LCDM& 3& 10,777,216& 6,000,000 \\ \hline
LiveJournal& 25,327& 121,625& 100,000 \\ \hline
MNIST& 786& 60,000& 10,000 \\ \hline
MovieLens& 51& 3,706& 6,040 \\ \hline
Netflix& 51& 17,770& 480,189 \\ \hline
OptDigits& 64& 1,347& 450 \\ \hline
Pall7& 7& 100,841& 100,841 \\ \hline
Physics& 78& 112,500& 37,500 \\ \hline
PSF& 2& 3,056,092& 3,056,092 \\ \hline
SJ2& 2& 50,000& 50,000 \\ \hline
U-Random& 20& 700,000& 300,000 \\ \hline
Y!-Music& 51& 624,961& 1,000,990 \\ \hline
\end{tabular}
\end{small}
\caption{\textbf{Datasets used for evaluation:} The dimensionality $\dims$ and number of points in the reference set $S$ and the set of queries $V$.}
\label{tab:datasets}
\end{table}
\begin{figure}
\centering
\includegraphics[width=0.45\columnwidth, angle = -90]{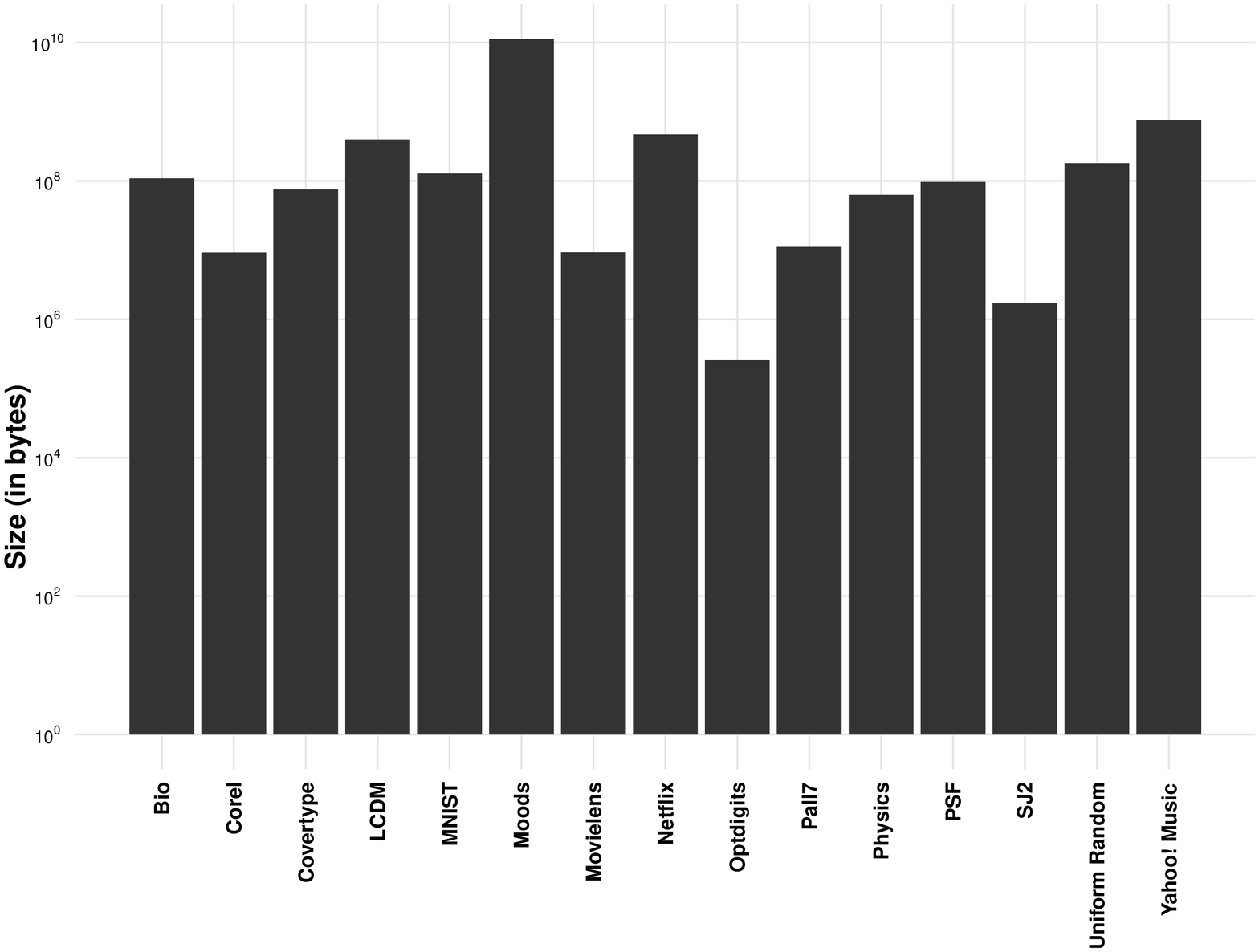}
\caption{\textbf{Total dataset sizes (in bytes):} The combined sizes of the reference set and the query set for each data set are presented in this figure.}
\label{fig:data_total_plot}
\end{figure}
%\begin{figure}[!htb]
%\centering
%\includegraphics[width=0.85\columnwidth, angle=-90]{./dataset_split_info}
%\caption{Dataset sizes (in bytes): The sizes of the reference set, query set and the combined sizes of each data set are presented in this figure.}
%\label{fig:data_split_plot}
%\vskip -6pt
%\end{figure}
\paragraph{Tree Construction Times.} The tree-building procedure is extremely efficient. We present the tree construction times in table \ref{tab:building-time} and contrast them with the runtime of the linear search algorithm (Alg. \ref{alg:linear-search}). For some of the larger data sets, the extrapolated runtime of Alg. \ref{alg:linear-search} is reported. In the last column, we present the ratio of the tree construction times with the runtimes of Alg. \ref{alg:linear-search}. For algorithm \ref{alg:single-tree-search} \& \ref{alg:dual-exact-search}(B), the tree construction involves building one and two ball-trees respectively. For algorithm \ref{alg:dual-exact-search}(C), the queries are normalized to have unit length for convenience since the norms of the queries do not affect the answers (equation \ref{eq:maxip_mod}). Following the query normalization, two trees are built. We include the query normalization in the tree construction time for completeness. This is the reason for the significant difference between construction times for algorithm \ref{alg:dual-exact-search}(B) and \ref{alg:dual-exact-search}(C). 

The numbers in the last column of table \ref{tab:building-time} (R) show how small the construction times are with respect to the actual linear search. The highest ratio is 0.15 for the OptDigits dataset. This implies that any speedup over 1.18 at search time is enough to compensate for the tree construction time. For most of the datasets, this ratio is much lower. Moreover, this tree building cost is a one time cost. Once the tree is built, it can be used for searching the dataset multiple times.

\begin{table}[t]
\centering
\begin{small}
\begin{tabular}{|l|c|c|c|c|c|} \hline
\textsc{Dataset} & Alg.\ref{alg:single-tree-search} & Alg.\ref{alg:dual-exact-search}(B) & Alg.\ref{alg:dual-exact-search}(C) & Alg.\ref{alg:linear-search} & R(\%) \\ \hline
%MovieLens@10 & 7.6 & 3.8 & 4.2 \\ \hline
%MovieLens@25 & 3.1 & 1.7 & 2.2 \\ \hline
Bio & 4.3 & 5.7 & 10.6 & 4,028 & 0.25\\ \hline
Corel & 0.2 & 0.27 & 0.66 & 43 & 1.5\\ \hline
Covertype & 5.5 & 7.2 & 14.8 & 14,885 & 0.1\\ \hline
LCDM & 36.7 & 56.46 & 99.3 & 1,984,200 & 0.005\\ \hline
LiveJournal & 2223 & 4073 & 4745 & 517,194 & 0.92\\ \hline
MNIST & 8.06 & 9.1 & 11.38 & 817 & 1.5\\ \hline
MovieLens & 0.03 & 0.08 & 0.27 & 4.62 & 6\\ \hline
Netflix & 0.2 & 8.27 & 33.5 & 1,878 & 1.7\\ \hline
OptDigits & 0.01 & 0.012  & 0.022 & 0.135 & 15\\ \hline
Pall7 & 0.26 & 0.52 & 1.4 & 364 & 0.4\\ \hline
Physics & 2.33 & 3.0 & 5.8 & 1,114 & 0.5\\ \hline
PSF & 9.06 & 18.1 & 34.95 & 282,514 & 0.01\\ \hline
SJ2 & 0.1 & 0.2 & 0.46 & 75 & 0.6\\ \hline
U-Rand & 4.94 & 6.9 & 15.64 & 26,586 & 0.6 \\ \hline
Y! Music & 9.72 & 28.85 & 112.5 & 137,306 & 0.08\\ \hline
\end{tabular}
\end{small}
\label{tab:building-time}
\caption{\textbf{Tree construction time (in seconds) contrasted with the linear search time (in seconds).}}
\end{table}

\paragraph{Search Efficiency.} The speedups over linear search are presented in Table \ref{tab:speedups}. We have reported every dataset we evaluated our algorithm on. Overall, the speedup numbers vary from as low as $1.13$ for the OptDigits dataset to over $10^5$ ($4$ orders of magnitude) for the LCDM and the PSF dataset. An important thing to note here is that for datasets with low speedup (below an order of magnitude) with Alg. \ref{alg:single-tree-search}, the speedup numbers for all three algorithms were pretty low and fairly comparable for all three algorithms. However, even a speedup of 2 is pretty significant in terms of absolute times. For example, for the Yahoo! music dataset, a search speedup of mere $2$ with a tree construction time of $120$ seconds gives a saving of 19 hours of computation time. For most datasets with a high value of speedup for Alg. \ref{alg:single-tree-search}, the speedups for the dual-tree algorithms are also very high. 

\begin{table}[t]
\centering
\begin{small}
\begin{tabular}{|l|c|c|c|} \hline
\textsc{Dataset} & Alg.\ref{alg:single-tree-search} & Alg.\ref{alg:dual-exact-search}(B) & Alg.\ref{alg:dual-exact-search}(C) \\ \hline
%MovieLens@10 & 7.6 & 3.8 & 4.2 \\ \hline
%MovieLens@25 & 3.1 & 1.7 & 2.2 \\ \hline
Bio & 7,059.62 & 6.55 & 273.52 \\ \hline
Corel & 14.27 & 17.38 & 7.68 \\ \hline
Covertype & 927.51 & 10.05 & 773.34\\ \hline
LCDM & 29,526 & 1,327 & 101,950 \\ \hline
LiveJournal & 28.04 & 10.42 & 15.45 \\ \hline
MNIST & 2.61 & 2.22 & 2.5 \\ \hline
% Mood & 2.00 & 0.923 & 4.18 \\ \hline
MovieLens & 2.23 & 1.36 & 1.67 \\ \hline
Netflix & 1.98 & 1.92 & 1.84 \\ \hline
OptDigits & 1.13 & 1.10 & 1.10 \\ \hline
Pall7 & 1,020 & 23.14 & 2,285 \\ \hline
Physics & 4.93 & 4.0 & 4.08 \\ \hline
PSF & 61,502 & 96,570 & 125,800 \\ \hline
SJ2 & 544 & 190 & 767 \\ \hline
U-Rand & 3.76 & 3.18 & 3.28 \\ \hline
Y!-Music & 2.11 & 2.09 & 2.16 \\ \hline
\end{tabular}
\end{small}
\label{tab:speedups}
\caption{\textbf{Speedups over linear search for $k=1$.}}
\end{table}

\begin{figure*}[!htb]
\centering
\includegraphics[width=0.85\textwidth, angle = -90]{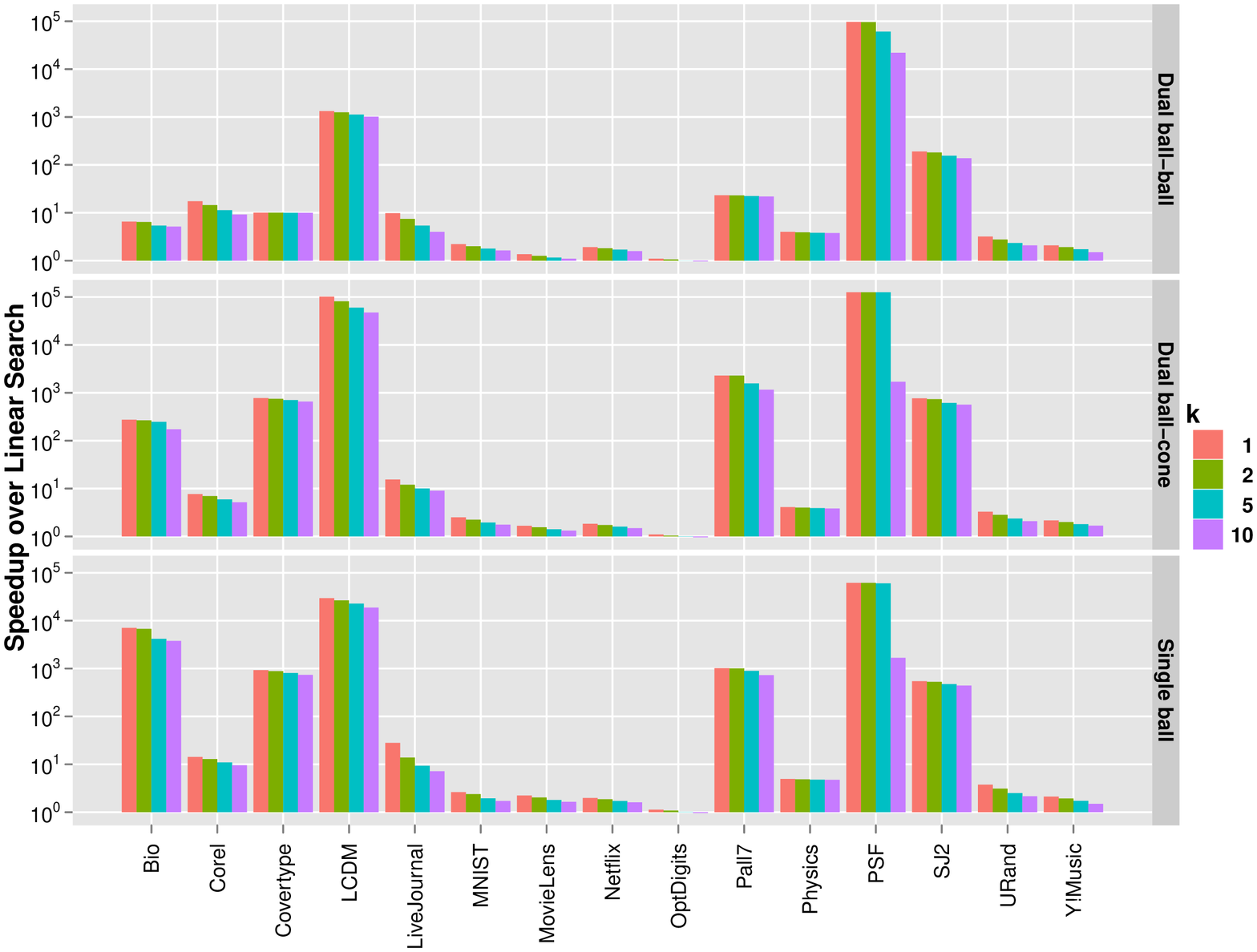}
\vspace{-0.25in}
\caption{\textbf{Speedups over linear search for $k = 1,\ 2,\ 5\ \&\ 10$.}}
\label{fig:speedups_k}
\end{figure*}

There are three important things to note here. Firstly, the dual-tree algorithms (Alg. \ref{alg:dual-exact-search}) do not perform very well if the single-tree algorithms (Alg. \ref{alg:single-tree-search}) does not have a high speedup. This is mostly because the tree is unable to find tight bounds and hence has to travel every branch. The dual-tree scheme loosens the bound to amortize the traversal cost over multiple queries. But if the bounds are bad for algorithm \ref{alg:single-tree-search}, the bounds for the dual-tree are much worse. Hence, the dual-tree algorithm does not show any significant speedup. Secondly, the dual-tree algorithm (especially Alg. \ref{alg:dual-exact-search}(C)) starts outperforming the single-tree algorithm significantly when the set of queries is really large. This is a usual behavior for dual-tree algorithms. The query set has to be large enough for the gains from the amortization of query traversal of the reference tree ({\em RTree}) to outweigh the computational cost of traversing the query-tree ({\em QTree}) itself. Finally, the dual-tree algorithm with ball-trees for the query set is generally significantly slower than the dual-tree with a cone-tree for the queries. There are possibly two possible reasons for that -- (i) The cones provide a tighter indexing of the queries than balls. A single cone can be used to index points in multiple balls which lie in the same direction but have varying norms. (ii) The upper bound for $\mathbf{MIP}(Q,T)$ in equation \ref{eq:cosine_bnd} is fairly loose. We do provide two ways of obtaining tighter bounds in the Appendix, but we have not yet evaluated the algorithm with the new bounding techniques.

We also consider the general problem of obtaining the points in the set $S$ with the $k$ highest inner-product with the query $q$. This is analogous to the $k$-nearest neighbor search problem. We present the speedups of our algorithms over linear search for $k = 1,\ 2,\ 5\ \&\ 10$ in figure \ref{fig:speedups_k}.

\vspace{-0.15in}
\section{Max-kernel Operation with General Kernel Functions}
\label{sec:max_kernel}
\vspace{-0.05in}
In this section, we show a method to apply the proposed algorithms in a inner-product space where the inner-products are defined by a kernel function, but it is not possible to explicitly represent the points in the $\varphi$-space. 

Without an explicit representation, the tree construction has to be modified since there would be no explicit representation of the mean of a set. For a tree node $T$ with the set of point $T.S$, the mean in $\varphi$-space is defined as  $$\mu = \frac{1}{|T.S|}\sum_{p \in T.S} \varphi(p).$$ $\mu$ might not have an explicit representation, but it is possible to compute inner products with $\mu$ as follows: $$\ip{\mu}{\varphi(q)} = \frac{1}{|T.S|} \sum_{p \in T.S} \mathcal{K}(q, p).$$
However, this computation is possibly very expensive (as opposed to the operation in equation \ref{eq:mip_ball_bnd} which is equivalent to a single inner-product). Instead of picking the mean of the set in the $\varphi$-space as the center of the ball, we propose picking the point in the $\varphi$-space which is closest to the mean $\mu$ as the new center. So the new center $p_c$ is given by: 
\begin{eqnarray}
p_c & = & \arg \min_{r \in T.S} \normsq{\varphi(r) - \mu} \nonumber \\ 
% & = & \arg \min_{r \in T.S} ||\varphi(r)||^2 - 2 \ip{\mu}{\varphi(r)} \nonumber \\
\label{eq:varphi-space-center} & = & \arg \min_{r \in T.S} \mathcal{K}(r,r) - \frac{2}{|T.S|} \sum_{r' \in T.S} \mathcal{K}(r', r).
\end{eqnarray}
This operation is quadratic in computation time, but is done at the preprocessing phase to provide efficiency during the search phase. Given this new center $p_c$, we can compute the radius $R_p$ of the ball enclosing the set $T.S$ as follows:
\begin{eqnarray}
R_p^2 & = & \max_{r \in T.S} \normsq{\varphi(r) - \varphi(p_c)} \nonumber \\
\label{eq:varphi-space-radius} & = & \max_{r \in T.S} \mathcal{K}(p_c, p_c) +  \mathcal{K}(r, r) - 2 \mathcal{K}(r, p_c).
\end{eqnarray}

Now given this method of choosing the center and evaluating the radius, a ball-tree can be built in the $\varphi$-space using Algorithm \ref{alg:ball-tree-construction} without ever requiring the explicit representation of the points. Given this ball-tree, the equation \ref{eq:mip_ball_bnd} in theorem \ref{thm:single-ball-bnd} can be modified to this situation as follows: 
\begin{equation} \label{eq:mip_varphi_space_ball_bnd}
\mathbf{MIP}(q, T) = \mathcal{K}(q, p_c) + R_p \sqrt{\mathcal{K}(q,q)},
\end{equation}
where $p_c$ is defined in equation \ref{eq:varphi-space-center} and $R_p$ is defined in equation \ref{eq:varphi-space-radius}. Computing this upper bound is equivalent to a single kernel function evaluation ($\mathcal{K}(q,q)$ be pre-computed before searching the tree). Using this upper bound, the tree-search algorithm (Alg. \ref{alg:single-tree-search}) can be performed in $\varphi$-space without any explicit representation of the points. We will present the evaluation of this method in the longer version of the paper.

Using the same principles, the dual-tree algorithm (Alg. \ref{alg:dual-exact-search}) can also be applied to the $\varphi$-space without any explicit representation of the points. For the dual-tree with ball-tree for the queries, the upper bound on the maximum inner-product between queries in node $Q$ and points in node $T$ in theorem \ref{thm:ball-ball-bound} becomes:
\begin{equation} \label{eq:mip_varphi_space_ball_ball_bnd} \mathbf{MIP}(Q,T) = \mathcal{K}(q_c, p_c) + R_p R_q + R_p \sqrt{\mathcal{K}(q_c, q_c)} + R_q \sqrt{\mathcal{K}(p_c, p_c)}, 
\end{equation}
where $p_c$ and $q_c$ are the chosen ball centers in the $\varphi$-space with radius $R_p$ and $R_q$ respectively.

For queries indexed in a cone-tree, the central axis of the cone can be the point in the $\varphi$-space making the smallest angle with the mean of the set in the $\varphi$-space. Since the queries are supposed to be normalized in the $\varphi$-space, for a query tree node $Q$, the mean of the set $Q.S$ is supposed to be:
$$\mu = \frac{1}{|Q.S|}\sum_{q \in Q.S} \frac{\varphi(q)}{\norm{\varphi(q)}}.$$
So the new central axis $q_c$ of the cone is given by:
\begin{eqnarray}
q_c & = & \arg \max_{q \in Q.S} \frac{\ip{\mu}{\varphi(q)}}{\norm{\mu}\norm{\varphi(q)}} \nonumber \\
 & = & \arg \max_{q \in Q.S} \frac{\sum\limits_{q'\in Q.S} \frac{\mathcal{K}(q', q)}{\sqrt{\mathcal{K}(q', q')}}}{\sqrt{\mathcal{K}(r,r)}}.
\end{eqnarray}
Again, this computation is quadratic in the size of the dataset, but provides efficiency during search time. The cosine of half the aperture of the cone is now given by:
\begin{eqnarray}
\cos \omega_q & = & \min_{q \in Q.S} \frac{\ip{\varphi(q_c)}{\varphi(q)}}{\norm{\varphi(q_c)} \norm{\varphi(q)}} \nonumber \\
& = & \min_{q \in Q.S} \frac{\mathcal{K}(q_c, q)}{\sqrt{\mathcal{K}(q_c, q_c) \mathcal{K}(q,q)}}.
\end{eqnarray}
Given $q_c$ and $\omega_q$, the upper bound in theorem \ref{thm:cone-ball-bnd} for a cone-tree node $Q$ of queries and a ball-tree node $T$ of reference points is given by:
\begin{equation}\label{eq:mip_varphi_space_cone_ball_bnd}
\mathbf{MIP}(Q,T) = \sqrt{\mathcal{K}(p_c, p_c)} \cos (\{|\phi| - \omega_q\}_+) + R_p, 
\end{equation}
where $\phi$ is defined as:
$$ \cos \phi = \frac{\mathcal{K}(p_c, q_c) }{\sqrt{\mathcal{K}(q_c, q_c) \mathcal{K}(p_c, p_c)}}.$$
This bound is very efficient to compute as it only requires a single kernel function evaluation (the terms $\mathcal{K}(p_c, p_c)$ and $\mathcal{K}(q_c, q_c)$ can be pre-computed and stored in the trees).
\vspace{-0.15in}
\section{Conclusion}
\vspace{-0.05in}
\label{sec:conclusions}
We consider the general problem of maximum inner-product search and present three novel methods to solve this problem efficiently. We use the tree data structure and present a branch-and-bound algorithm for maximum inner-product search. We further extend it to the case where the set of queries is very large. We evaluate the proposed algorithms with a variety of datasets and exhibit their computational efficiency.

A theoretical analyses of these proposed algorithms would give us a better understanding of the computational efficiency of these algorithms. 
%
%However, note that the proposed algorithms benefit from the ability of the tree data structures to index the data and hence do suffer from the ``curse of dimensionality''. Investigating more recent data structures like cover trees \cite{beygelzimer2006cover} and RP-trees \cite{dasgupta2008random} is a possible future direction. Moreover, a theoretical analyses of these proposed algorithms would give us a better understanding of the computational efficiency of these algorithms. 
We do not have any rigorous runtime bounds for our algorithm and it would be part of our future work.

\bibliographystyle{abbrv}
\bibliography{../pram_bib_col,../nns,../nbd,../ml,../manifold,../de,../applications}

\appendix
\section{Tighter Bounds with Optimization}
%\begin{citemize}
%\item[->] tighter bounding (in appendix) -- needs cleanup 
%\end{citemize}

In this section, we present two ways to get a tighter bound on equation \ref{eq:max_max} with respect to $\theta_p$ and $\theta_q$. The maximum inner product bound $\mathbf{MIP}(Q,T)$ between two balls is given as:
\begin{equation}
\label{eq:max_R1}\ip{q^*}{p^*}  \leq  \max_{\theta_p, \theta_q, r_p, r_q} \ip{q_0}{p_0} + r_p r_q \cos(\phi-(\theta_p + \theta_q))  + r_p||q_0|| \cos (\phi-\theta_p) + r_q ||p_0|| \cos (\phi - \theta_q).
\end{equation}

\subsection{Two-variable Optimization}
Assuming that 
$$| \phi - (\theta_p + \theta_q)|  \leq  \frac{\pi}{2}, |\phi - \theta_p|  \leq  \frac{\pi}{2}, |\phi - \theta_q|  \leq  \frac{\pi}{2},$$
we can say that:
\begin{eqnarray}
\label{eq:max_R2}\ip{q^*}{p^*} & \leq & \max_{\theta_p, \theta_q} \ip{q_0}{p_0} + R_p R_q \cos(\phi-(\theta_p + \theta_q)) + R_p||q_0|| \cos (\phi-\theta_p) + R_q ||p_0|| \cos (\phi - \theta_q)  \\ 
& = & f(\theta_p, \theta_q).
\end{eqnarray}
%
%
%However, the bound (equation \ref{eq:mip_ball_ball_bnd}) obtained from theorem \ref{thm:ball-ball-bound} appears to be very loose. In this section, we provide possible ways of tightening that bound by using two-variable or one variable optimization.

Now $\frac{\partial f(\theta_p, \theta_q)}{\partial \theta_p} = 0$ and $\frac{\partial f(\theta_p, \theta_q)}{\partial \theta_q} = 0$ gives us the following optimality conditions:
\begin{equation}
\label{eq:opt_theta_p}
\frac{\sin(\phi - (\theta_p+\theta_q))}{\sin (\phi-\theta_p)} = -\frac{||q_0||}{R_q},
\end{equation}
\begin{equation}
\label{eq:opt_theta_q}
\frac{\sin(\phi - (\theta_p+\theta_q))}{\sin (\phi-\theta_q)} = -\frac{||p_0||}{R_p}.
\end{equation}
The second order conditions are the following:
\begin{eqnarray}
\label{opt_2nd}
\frac{\partial^2 f(\theta_p,\theta_q)}{\partial \theta_p^2} & = & - R_p R_q \cos(\phi - (\theta_p+\theta_q)) - R_p ||q_0|| \cos(\phi - \theta_p), \\
\frac{\partial^2 f(\theta_p,\theta_q)}{\partial \theta_q^2} & = & - R_p R_q \cos(\phi - (\theta_p+\theta_q)) - R_q ||p_0|| \cos(\phi - \theta_q),\\
\frac{\partial^2 f(\theta_p,\theta_q)}{\partial \theta_p \partial \theta_q} & = & - R_p R_q \cos(\phi - (\theta_p+\theta_q)),
\end{eqnarray}
which are all $< 0$ for the stated range of $\phi, \theta_p$ and $\theta_q$. So the optimal values obtained from the optimality conditions (equations \ref{eq:opt_theta_p} \& \ref{eq:opt_theta_q}) correspond to the maximum. However, the optimality conditions do not have an analytic solution for $\theta_p$ and $\theta_q$. Hence, any efficient optimization algorithm can be used to solve $\max_{\theta_p, \theta_q} f(\theta_p, \theta_q)$ in the specified range. 
%
%Now, let $\gamma = \cos(\phi-(\theta_p+\theta_q))$. Then, using the optimality conditions in equations \ref{eq:opt_theta_p} \&~\ref{eq:opt_theta_q} and replacing it in $f(\theta_p,\theta_q)$, we get
%\begin{eqnarray}
%\label{eq:2var_opt} 
%f(\theta_p,\theta_q) & = & \ip{q_0}{p_0} + R_p R_q \gamma \nonumber \\ 
%& &  \ \ +  \sqrt{R_q^2 ||p_0||^2 - R_p^2 R_q^2 (1 - \gamma^2)} \nonumber \\
%& &  \ \ + \sqrt{R_p^2 ||q_0||^2 - R_p^2 R_q^2 (1 - \gamma^2)}.
%\end{eqnarray}
%Using the loose bound that $\gamma = \cos(\phi - (\theta_p + \theta_q)) \leq 1$, you get the following upper bound:
%\begin{eqnarray}
%\label{eq:l_ball_bnd}
%\max\limits_{p \in B_{p_0}^{R_p}, q \in B_{q_0}^{R_q}} \ip{q}{p} &  \leq  & f(\theta_p, \theta_q) \nonumber \\
%& \leq & \ip{q_0}{p_0} + R_p R_q \nonumber \\
%& & \ \ + R_q ||p_0|| + R_p ||q_0||.
%\end{eqnarray}
%
%There does not seem to be any other closed form analytical form for $\max\limits_{\theta_p, \theta_q} f(\theta_p,\theta_q)$. But you can use (crude, slow) grid search over the span of $(\theta_p, \theta_q)$ or (faster) coordinate descent (or something else) to get tighter ball-to-ball maximum inner-product bounds.
%
%\paragraph{Remark.} The loose pruning rule in Equation \ref{eq:l_ball_bnd} has been implemented and does not give any promising results (and most of the times it is slower than the single tree).
\subsection{One-variable Optimization}
Another approach is the following:
\begin{equation} \label{eq:broken_max}
\begin{split}
\max_{\theta_p,\theta_q, r_p, r_q}  \ip{p_0}{q_0} + r_p r_q \cos(\phi-(\theta_p + \theta_q)) + r_p||q_0|| \cos (\phi-\theta_p) + r_q ||p_0|| \cos (\phi - \theta_q) \ \ \ \ \ \ \ \ \ \ \ \ \ \ \ \ \ \  \\
\ \ \ \ \ \ \ \ \ \ \ \ \ \ \ \ \ \  \leq \max_{\theta_q, r_p, r_q} \max_{\theta_p} p_0^T q_0 + r_p r_q \cos(\phi-(\theta_p + \theta_q)) + r_p||q_0|| \cos (\phi-\theta_p) + r_q ||p_0|| \cos (\phi - \theta_q),
\end{split}
\end{equation}
Since for fixed $\theta_q$, $\omega_q$ is fixed. And for a fixed $\omega_q$, using the single-tree bounding, $\theta_p = (\phi - \omega_q)$. Making this substitution in equation \ref{eq:broken_max}, we get the following optimization task:
\begin{eqnarray}
\ip{p^*}{q^*} & \leq &  \max_{\theta_q, r_p, r_q} \ip{p_0}{q_0} + r_p ||q^*|| + r_q ||p_0|| \cos(\phi - \theta_q) \\
\label{eq:1var_opt}& \leq & \max_{\theta_q} \ip{p_0}{q_0} + R_q ||p_0|| \cos(\phi - \theta_q) + R_p \sqrt{||q_0||^2 + R_q^2 + 2 R_q ||q_0|| \cos \theta_q} \\
&  =  & f(\theta_p), \nonumber
\end{eqnarray}
where the second inequality comes from the assumption that $$|\theta_q| \leq \frac{\pi}{2}, |\phi - \theta_q| \leq \frac{\pi}{2},$$ and the fact that $r_p \leq R_p,\ r_q \leq R_q$. 
%Using the loose bound of $\cos(\cdot)\leq 1$, you arrive at the same bound as equation \ref{eq:l_ball_bnd}. On the other hand, you can perform a one-dimensional optimization over $\theta_q$ (in the range $[-\pi/2, \pi/2]$) to obtain a tighter upper bound.

The first-order optimality condition gives us the following:
\begin{equation*} \label{eq:opt2_theta_p}
R_p \frac{R_q ||q_0|| \sin \theta_p}{\sqrt{||q_0||^2 + R_q^2 + 2 R_q ||q_0|| \cos \theta_q}} = R_q ||p_0|| \sin (\phi - \theta_q),
\end{equation*}
while the second-order derivative is given by:
\begin{equation*} \label{eq:opt2_2nd}
-R_q R_p ||q_0||  \frac{\cos \theta_p (||q_0||^2 + R_q^2 +  R_q ||q_0|| \cos \theta_q) + R_q ||q_0||}{\left(||q_0||^2 + R_q^2 + 2 R_q ||q_0|| \cos \theta_q \right)^{3/2}} + R_q ||p_0|| \cos(\phi - \theta_q),
\end{equation*}
which is always $< 0$ implying that the optimal $\theta_p$ is the maximum even though the equation \ref{eq:opt2_theta_p} does not give an analytic solution for $\theta_p$. An efficient optimization algorithm can be used to solve this one dimensional optimization problem $\max_{\theta_p} f(\theta_p)$ to obtain tight bounds for $\textbf{MIP}(Q,T)$. 

%\paragraph{Remark.} This has been implemented without any 
promising result.

\end{document}